\documentclass[11pt,reqno]{amsart}
\pdfoutput=1 
\usepackage{amsmath,amssymb}
\usepackage{graphicx}
\usepackage[pagebackref, colorlinks = true, linkcolor = blue, urlcolor  = blue, citecolor = red]{hyperref}
\numberwithin{equation}{section}
\usepackage{multirow}
\usepackage{subfigure}

\usepackage[margin=0.8in]{geometry}

\renewcommand{\P}{\mathbb{P}}

\newcommand{\N}{\mathbb{N}}
\newcommand{\M}{\mathcal{M}}
\newcommand{\C}{\mathbb{C}}
\newcommand{\U}{\mathcal{U}}
\newcommand{\iy}{\infty}
\newcommand{\K}{\mathcal{K}}
\newcommand{\R}{\mathbb{R}}
\newcommand{\E}{\mathbb{E}}
\newcommand{\D}{\mathcal{D}}
\newcommand{\norm}[1]{\left\Vert #1\right\Vert}
\newcommand{\normt}[1]{\left\Vert #1\right\Vert_{(t)}}
\newcommand{\scalar}[2]{\langle #1 , #2\rangle}
\newcommand{\ket}[1]{| #1 \rangle}
\newcommand{\bra}[1]{\langle #1 |}

\renewcommand{\H}{\mathcal{H}}
\newcommand{\ol}{\overline}

\DeclareMathOperator{\Gr}{Gr}
\DeclareMathOperator{\trace}{Tr}

\DeclareMathOperator{\I}{I}
\DeclareMathOperator{\id}{id}
\DeclareMathOperator{\Wg}{Wg}
\DeclareMathOperator{\Mob}{Mob}
\DeclareMathOperator{\Rem}{Rem}

\newtheorem{theorem}{Theorem}[section]
\newtheorem{definition}[theorem]{Definition}
\newtheorem{proposition}[theorem]{Proposition}
\newtheorem{corollary}[theorem]{Corollary}
\newtheorem{lemma}[theorem]{Lemma}
\newtheorem{remark}[theorem]{Remark}

\newtheorem{question}[theorem]{Question}

\begin{document}

\date{\today}

\title{Random matrix techniques in quantum information theory}

\author {Beno{\^i}t Collins}
\address{BC: 
Department of Mathematics, 
Kyoto University, Kyoto 606-8502 JAPAN
and
D\'epartement de Math\'ematique et Statistique, Universit\'e d'Ottawa,
585 King Edward, Ottawa, ON, K1N6N5 Canada, and CNRS} 
\email{collins@math.kyoto-u.ac.jp}

\author{Ion Nechita}
\address{IN: Zentrum Mathematik, M5, Technische Universit\"at M\"unchen, Boltzmannstrasse 3, 85748 Garching, Germany
 and CNRS, Laboratoire de Physique Th\'{e}orique, IRSAMC, Universit\'{e} de Toulouse, UPS, F-31062 Toulouse, France}
\email{nechita@irsamc.ups-tlse.fr}

\subjclass[2000]{}

\begin{abstract}
The purpose of this review article is to present some of the latest
developments using random techniques, and in particular,
random matrix techniques in quantum information theory.
Our review is a blend of a rather exhaustive review, combined
with more detailed examples -- coming from research projects
in which the authors were involved. 
We focus on two main topics, random quantum states and random quantum channels. We present results related to entropic quantities, entanglement of typical states, entanglement thresholds, the output set of quantum channels, and violations of the minimum output entropy of random channels.  
\end{abstract}

\maketitle

\tableofcontents

\section{Introduction}

Quantum computing was initiated in 1981 by Richard Feynman at a conference on physics and computation at the MIT, where he 
asked whether one can simulate physics on a computer. 
As for quantum information theory, it is somehow the backbone
of quantum computing, although it emerged independently in the 
60s and 70s (with, among others, Bell inequalities). In the last 20-30 years, it has witnessed a very fast  development, and it is now a major scientific field of its own.

In classical information theory, probabilistic methods have
been at the heart of the theory from its inception \cite{shannon1948mathematical}. 
In contrast, 
probabilistic methods have arguably played a less important role
in the infancy of quantum information theory -- although probability
theory itself was cast at the heart of the postulates
of quantum mechanics (Bell inequalities themselves are a probabilistic
statement). 
However, the situation has dramatically changed in the last 10-15 years,
and probabilistic techniques have nowadays proven to be very useful
in quantum information theory. Quite often, these probabilistic 
techniques are very closely related to problems in random matrix theory. 

Let us explain heuristically why random matrix theory is a natural 
tool for quantum information theory, by comparing with the use
of elementary probability in the concept of classical information. 
In classical information, the first -- and arguably one of the greatest --
success of the theory was to compute the relative volume of a typical set
with Shannon's entropy function. 
Here, the main probabilistic tool was the central limit theorem 
(more precisely, an exponential version thereof). 
The central limit theorem is a tool that is very well adapted 
to the study of product measures on the product of (finite) sets. 
In quantum information, sets and probability measures on theses
sets, as well as measurements, are all replaced by matrices, and
their non-commutative structure is central to quantum information
theory.
In addition, the use of `probabilistic techniques' in classical information
theory is not a goal \emph{per se}. It is a convenient mathematical tool to 
prove existence theorems, for which non-random proofs are much more
difficult to achieve. 
Incidentally, it is rather natural to wonder what a `typical' set looks like.

The situation is quite similar in quantum information theory: 
although there was no a priori need for random techniques, some
problems -- in particular the minimum output entropy additivity 
problem, which we discuss at length in this review -- did not have
an obvious non-random answer, therefore it became not only
natural, but also important, to consider how `typical' quantum objects
behave. This was arguably initiated by the paper
\cite{hayden2004randomizing}.
The results obtained in the first papers were of striking importance
in quantum information theory, and they pointed at the fact that
well established mathematical tools could be expected
to be very useful tools to solve problems in quantum information theory. 
All these problems were naturally linked with probability measures
on matrix spaces. The first tool that was recognized to play an 
important role was concentration of measure. 
But more generally, all techniques were 
connected to the theory of  random matrices.
Random matrix theory relies on a wide range of technical tools, and
concentration of measure is one of them among others. 

The first results using random techniques in quantum information theory
attracted a few mathematicians -- including the authors of this
review, see also the bibliography -- who undertook to 
apply systematically the state of the art in random matrix theory
in order to study questions in quantum information theory. 

Random matrix theory itself has a long history.
Although it is considered as a field of mathematics (or mathematical
physics), it was not born in mathematics, but rather in statistics and physics.
Wishart introduced the distribution that bears his name in the 1920's,
in order to explain the discrepancy between the eigenvalues of a measured
covariance matrix, and an expected covariance matrix. 
And Wigner had motivations from quantum
physics when he introduced the semi-circle distribution.
Since then, random matrix theory has played a role in many fields
of mathematics and science, including: 
\begin{itemize}
\item
Theoretical physics \cite{witten1986physics,kontsevich1994homological,mehta2004random}
\item
Combinatorics and algebraic geometry 
\cite{harer1986euler,okounkov2004random}
\item
Integrable systems and PDE 
\cite{tracy2002distribution}
\item
Complex analysis and Riemann-Hilbert problems 
\cite{kuijlaars2004riemann}
\item
Operator algebras \cite{voiculescu1992free}
\item
Telecommunication \cite{tulino2004random}
\item
Finance \cite{bouchaud2003theory}
\item
Number theory 
\cite{katz1999random}
\end{itemize}

The above list does certainly not exhaust the list of fields of application of random matrix theory, 
but quantum information theory is definitely one of the most
recent of them (and a very natural one, too!). 

Our goal in this paper is to provide an overview of a few important uses
that random matrix theory had in quantum information theory. 
Instead of being exhaustive, we chose to pick a few topics that look
important to us, and hopefully emblematic of the roles that 
random matrices could play in the future in QIT. Obviously,
our choices are biased by our own experience.

This paper is organized as follows
Section \ref{sec:background} provides some mathematical notation for quantum information. 
It is followed by Section \ref{sec:RMT-FPT}, that supplies background for random matrix and free probability theory. 
The remaining sections are a selection of applications of random matrix theory to quantum information, namely:
\ref{sec:entanglement}: Entanglement of random quantum states, \ref{sec:single-output}: properties of output of deterministic states (of interest) under quantum channels, \ref{sec:random-output}: study of \emph{all} outputs under specific random quantum channels, 
\ref{sec:additivity}: the solution to the MOE additivity problems, and finally \ref{sec:other}: other applications of RMT in quantum physics, and \ref{sec:questions}: a selection of open questions.

\section{Background on quantum information: quantum states and channels}\label{sec:background}

\subsection{Quantum states}

We denote the set of $d$-dimensional mixed quantum states (or density matrices) by $\mathcal D_d$:
\begin{equation}\label{eq:density-matrices}
\mathcal D_d := \{ \rho \in M_d^{sa}(\mathbb C) \, : \, \rho \geq 0 \text{ and } \mathrm{Tr}(\rho) = 1 \}.
\end{equation}

The collection of density matrices $\mathcal D_d$ is naturally associated to the Hilbert
space $\C^d$. One of the fundamental postulates of quantum mechanics is that 
two disjoint systems can be studied together by taking their Hilbert tensor product. 
For example, a system with state space $\C^{d_1}$ and another one with
state space $\C^{d_2}$ are studied together with the state space
$\C^{d_1}\otimes\C^{d_1}$.
In particular, the density matrices have the structure of $\mathcal D_{d_1d_2}$.
Of particular interest is the subset 
$$\mathcal{SEP}_{d_1,d_2} := \mathrm{conv}(\mathcal D_{d_1}\otimes \mathcal D_{d_2}) \subseteq \mathcal D_{d_1d_2}.$$
This convex compact subset is interpreted as the collection of all `classical' density matrices on the 
bipartite state. Unless $d_1=1$ or $d_2=1$ it is a strict subset of $\mathcal D_{d_1d_2}$.
It is called the collection of \emph{separable states}. States which are not
separable are called \emph{entangled}.
The study of entangled states is one of the cornerstones of quantum information theory. As a first (and paramount) example of entangled state, consider the qubit Hilbert space $\mathbb C^2$ endowed with an orthonormal basis $\{e_1,e_2\}$, and the state
$$\mathcal D_{4} \ni\Omega_2 = \omega_2 \omega_2^*, \quad \text{ where } \quad \omega_2 = \frac{1}{\sqrt 2}(e_1 \otimes e_1 + e_2 \otimes e_2),$$
called the maximally entangled qubit state. More generally, the maximally entangled state of two qudits is
\begin{equation}\label{eq:maximally-entangled-state}
\mathcal D_{d^2} \ni \Omega_d = \omega_d \omega_d^*, \quad \text{ where } \quad \omega_d = \frac{1}{\sqrt d}\sum_{i=1}^d e_i \otimes e_i.
\end{equation}

\subsection{Entropies}

As in classical information theory \cite{shannon1948mathematical}, entropic quantities play a very important role in quantum information theory. We define next the quantities of interest for the current work. 
Let $\Delta_k = \{x \in \R_+^k \, | \, \sum_{i=1}^k x_i = 1\}$ be the $(k-1)$-dimensional 
probability simplex. For a positive real number $p>0$, define the \emph{R\'enyi entropy of order $p$} of a probability vector $x \in \Delta_k$ to be
\[H^p(x) = \frac{1}{1-p}\log\sum_{i=1}^k x_i^p.\]
Since $\lim_{p \to 1} H^p(x)$ exists, we define the \emph{Shannon entropy} of $x$ to be
this limit, namely:
\[H(x) = H^1(x) = -\sum_{i=1}^k x_i \log x_i.\]
We extend these definitions to density matrices by functional calculus:
\begin{align*}
H^p(\rho) &= \frac{1}{1-p}\log \trace \rho^p;\\
H(\rho) &= H^1(\rho) = -\trace \rho \log \rho.
\end{align*}

\subsection{Quantum Channels}

In Quantum Information Theory, a \emph{quantum channel} is the most general transformation 
of a quantum system. Quantum channels generalize the unitary evolution of isolated quantum systems 
to \emph{open quantum systems}. Mathematically, we recall that a quantum channel is a linear 
completely positive trace preserving map $\Phi$ from $\M_n(\C)$ to itself. 
The trace preservation condition is necessary since quantum channels should map density matrices 
to density matrices. The complete positivity condition can be stated as
\[ \forall d \geq 1, \quad \Phi \otimes \I_d : \M_{nd}(\C) \to \M_{nd}(\C) \text{ is a positive map.}\]

The following two characterizations of quantum channels turn out to be very useful.

\begin{proposition}\label{prop:stinespring_kraus}
A linear map $\Phi : \M_n(\C) \to \M_n(\C)$ is a quantum channel if and only if one of the following two equivalent conditions holds.
\begin{enumerate}
\item \emph{(Stinespring dilation)} There exists a finite dimensional Hilbert space $\K = \C^{d}$, a density matrix $Y \in \M_{d}(\C)$ and an unitary operator $U \in \U(nd)$ such that
\begin{equation}\label{eq:Stinespring_form}
\Phi(X) = \trace_\K\left[ U(X \otimes Y) U^* \right], \quad \forall X \in \M_n(\C).
\end{equation}
\item \emph{(Kraus decomposition)} There exists an integer $k$ and matrices $L_1, \ldots, L_k \in \M_n(\C)$ such that
\[\Phi(X) = \sum_{i=1}^{k} L_i X L_i^* ,\quad \forall X \in \M_n(\C).\]

and
\[\sum_{i=1}^{k} L_i^* L_i = \I_n.\]
\item \emph{(Choi matrix)} The following matrix, called the \emph{Choi matrix} of $\Phi$
\begin{equation}\label{eq:def-Choi-matrix}
\mathcal M_{n^2}(\mathbb C) \ni C_\Phi = [\mathrm{id} \otimes \Phi](\Omega_d) = \sum_{i,j=1}^n E_{ij} \otimes \Phi(E_{ij})
\end{equation}
is positive-semidefinite. 
\end{enumerate}
\end{proposition}

It can be shown that the dimension of the ancilla space $\K$ in the Stinespring dilation theorem can be chosen $d = \dim \K = n^2$ and that the state $Y$ can always be considered to be a rank one projector. A similar result holds for the number of Kraus operators: one can always find a decomposition with $k=n^2$ operators. 

Going back to the entropic quantities, of special interest for the computation of capacities of quantum channels to transmit classical information are the following expressions, indexed by some positive real parameter $p$
$$H_{\min}^p(\Phi)=\min_{\rho \in \mathcal D_n} H^p (\Phi (\rho)).$$

\subsection{Graphical notation for tensors}
\label{sec:graphical-tensors}

Quantum states, tensors, and operations
between these objects (composition,
tensor product, applying a state through
a quantum channel, etc) can be efficiently represented
graphically. The leading idea is that a string in a diagram
means a tensor contraction. 
Many graphical theories for tensors and linear algebra computations have been developed in the literature \cite{penrose2005road,coecke2010quantum}.
Although they are all more or less equivalent, 
we will stick to the one introduced in 
\cite{collins2010random}, as it allows to compute the expectation
of random diagrams in a diagrammatic way
subsequently. For more details on this method,
we refer the reader to the paper \cite{collins2010random} and to other work which make use of this technique \cite{collins2011gaussianization,collins2010eigenvalue,collins2010randoma,collins2013area,fukuda2013partial,collins2013matrix,lancien2015extendibility}

In the graphical calculus matrices (or, more generally, tensors) are represented by boxes.
Each box has differently shaped symbols, 
where the number of different types of them equals that of different spaces (exceptions are mentioned bellow). 
Those symbols are empty (while) or filled (black), corresponding to primal or dual spaces. 
Wires connect these symbols, corresponding to tensor contractions. A diagram is a collection of such boxes and wires and corresponds to an element of an abstract element in a tensor product space.

Rather than going through the whole theory, we
focus on a few key examples. 

\begin{figure}
\includegraphics[width=2cm]{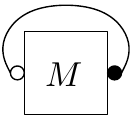}\quad\quad\quad
\includegraphics[width=4cm]{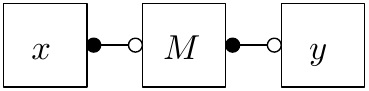}\quad\quad\quad
\includegraphics[width=2cm]{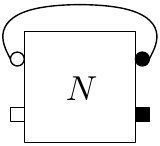}
\caption{Some simple diagrams}
\label{fig:simple}
\end{figure}

Suppose that each diagram in Figure \ref{fig:simple} comes equipped with two vector spaces $V_1$ and $V_2$ 
which we shall represent respectively by circle and square shaped symbols. 
In the first diagram, $M$ is a tensor (or a matrix, depending on which point of view we adopt) $M \in  V_1^* \otimes V_1$, and the wire applies the contraction $V_1^* \otimes V_1 \to \C$ to $M$. 
The result of the diagram $\D_a$ is thus $T_{\D_a} = \trace(M) \in \C$. 
In the second diagram, again there are no free decorations, hence the result is the complex number $T_{\D_b} = \scalar{y}{Mx}$. 
Finally, in the third example, $N$ is a $(2,2)$ tensor or a linear map $N \in \mathrm{End}(V_1 \otimes V_2, V_1 \otimes V_2)$. 
When one applies to the tensor $N$ the contraction of the couple $(V_1, V_1^*)$, 
the result is the partial trace of $N$ over the space $V_1$: $T_{\D_c} = \trace_{V_1}(N) \in \mathrm{End}(V_2, V_2)$.

\section{Background on random matrix theory and free probability}
\label{sec:RMT-FPT}

\subsection{Gaussian random variables}
\label{sec:Gaussian-random-variables}

The probability density of 
the \emph{normal distribution} is:
$$f(x \; | \; \mu, \sigma) = \frac{1}{\sigma\sqrt{2\pi} } \; e^{ -\frac{(x-\mu)^2}{2\sigma^2} }$$
Here, $\mu$ is the \emph{mean}. 
The parameter $\sigma$ is its 
\emph{standard deviation} with its variance then $\sigma^2$. A random variable with a Gaussian distribution is said to be normally distributed

Suppose $X$ and $Y$ are random vectors in $\R^k$ such that $(X,Y)$ is a $2k$-dimensional normal 
random vector. Then we say that the complex random vector
$Z = X + iY$
has the \emph{complex normal distribution}.

Historically the first ensemble of random matrices having been studied is the Wishart ensemble \cite{wishart1928generalised}, see \cite[Chapter 3]{bai2010spectral} or \cite[Section 2.1]{anderson2010introduction} for a modern presentation.

\begin{definition}\label{def:Wishart}
Let $X \in \mathcal M_{d \times s}(\mathbb C)$ be a random matrix with complex, standard, i.i.d.~ Gaussian entries. The distribution of the positive-semidefinite matrix $W = XX^* \in \mathcal M_d(\mathbb C)$ is called a \emph{Wishart distribution} of parameters $(d,s)$ and is denoted by $\mathcal W_{d,s}$. 
\end{definition}
The study of the asymptotic behavior of Wishart random matrices is due to Mar{\v{c}}enko and Pastur \cite{marcenko1967distribution}, while the strong convergence in the theorem below has
has been proved
by analytic tools such as 
determinantal point processes. 
Let us also record it as a direct
consequence of the much more general results  \cite{male2012norm}.

\begin{theorem}\label{thm:marchenko-pastur}
Consider a sequence $s_d$ of positive integers which behaves as $s_d \sim cd$ as $d \to \infty$, for some constant $c \in (0,\infty)$. Let $W_d$ be a sequence of positive-semidefinite random matrices such that $W_d \sim \mathcal W_{d,s_d}$. Then, the sequence $W_d$ converges strongly to the \emph{Mar{\v{c}}enko-Pastur  distribution} $\pi_c$ given by
\begin{equation}\label{eq:Marchenko-Pastur}
\pi_c=\max (1-c,0)\delta_0+\frac{\sqrt{(b-x)(x-a)}}{2\pi x} \; \mathbf{1}_{(a,b)}(x) \, dx,
\end{equation}
where $a = (1-\sqrt c)^2$ and $b=(1+\sqrt c)^2$.
\end{theorem}

The Mar{\v{c}}enko-Pastur distribution $\pi_c$ is sometimes called the \emph{free Poisson distribution}. 
We plotted in Figure \ref{fig:Marchenko-Pastur} its density in the cases $c=1$ and $c=4$.

\begin{figure}[htbp]
\begin{center}
\includegraphics[width=0.4\textwidth]{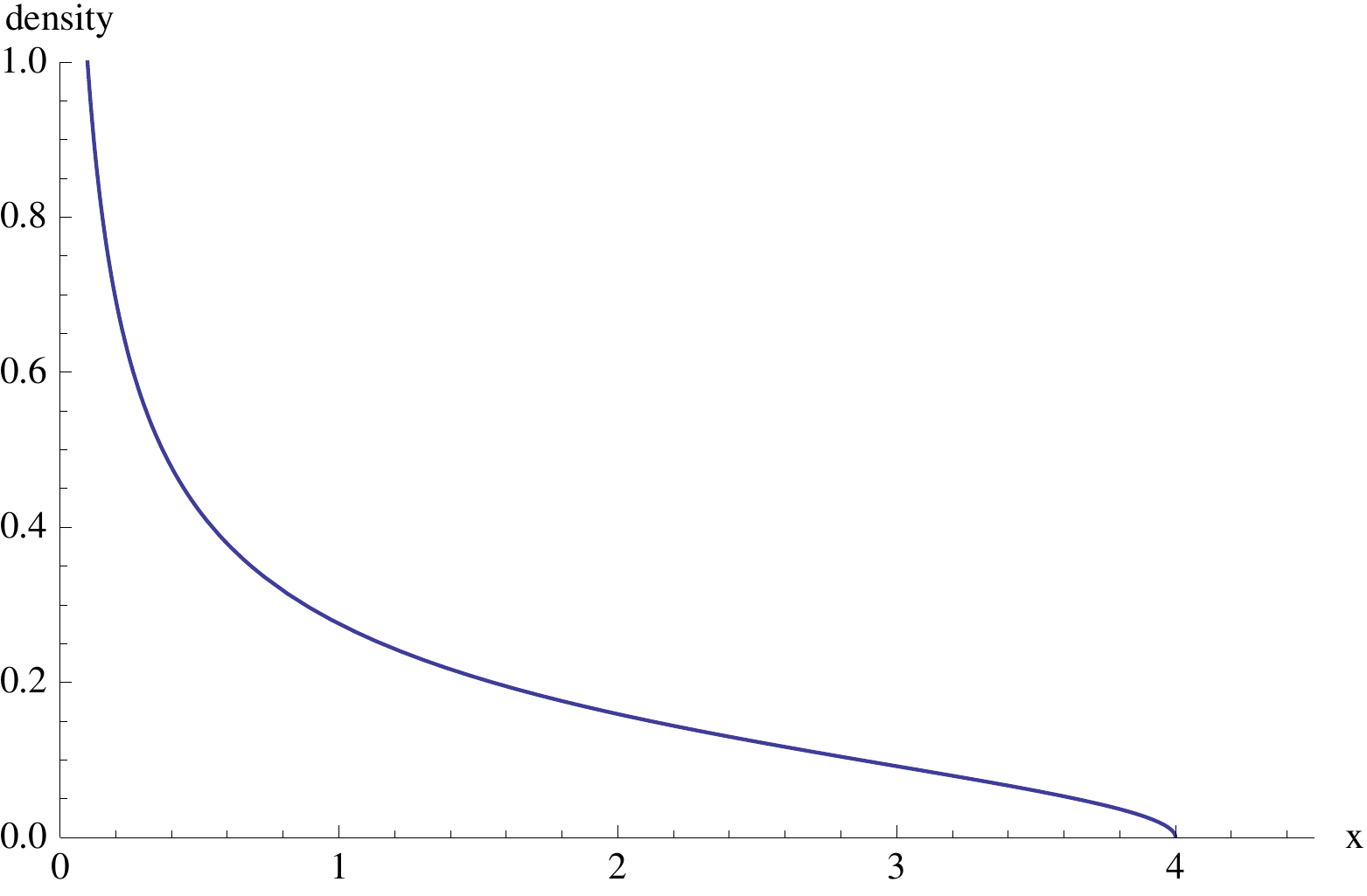} \quad \includegraphics[width=0.4\textwidth]{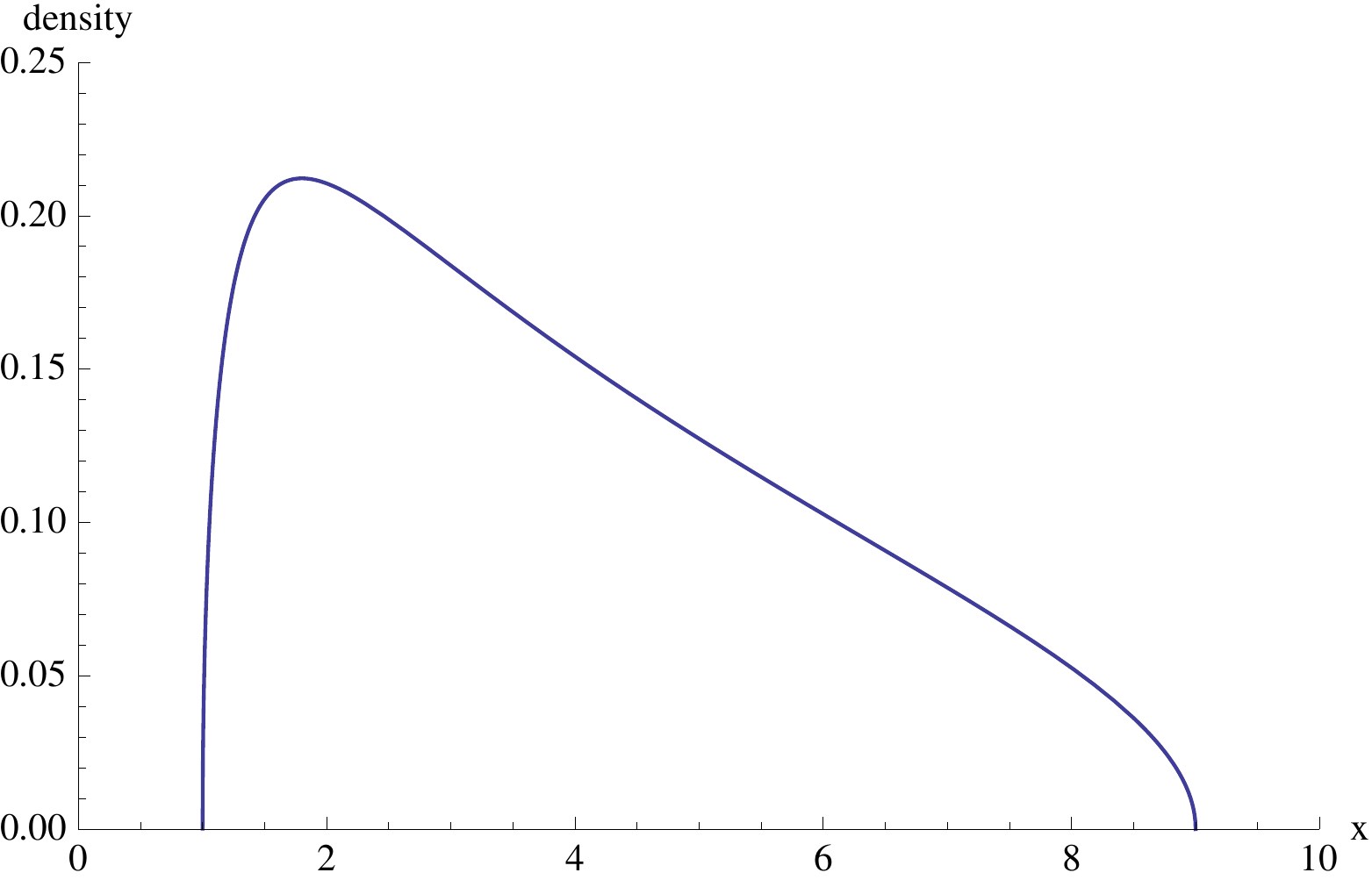}
\caption{The density of the Mar{\v{c}}enko-Pastur distributions $\pi_1$ (left) and $\pi_4$ (right).}
\label{fig:Marchenko-Pastur}
\end{center}
\end{figure}

The following theorem is the link between
combinatorics and probability theory
for Gaussian vectors: 
it allows to compute moments 
of any Gaussian vector thanks
to its covariance matrix. 

A \emph{Gaussian space} $V$ is a real vector space of random variables
having moments of all orders, with the property that each of these random variables
has centered Gaussian distributions. In order to specify the covariance information, such a Gaussian space comes with a positive symmetric bilinear form 
 $(x,y)\to \E[xy]$.
 Gaussian spaces are in one-to-one correspondence
with euclidean spaces.
In particular, the euclidean norm of a random variable
determines it fully (via its variance) and if two 
random variables are given, their joint distribution is
determined by their angle. 
The following is usually called the Wick Lemma:

\begin{theorem}\label{thm:wick-formula}
Let $V$ be a Gaussian space and $x_{1}, \ldots , x_k$ be elements in $V$.
If $k=2l+1$ then $\E[x_1\cdots x_k]=0$ and
if $k=2l$ then
\begin{equation}\label{eq:Wick}
\E[x_1\cdots x_k]=\sum_{\substack{p=\{\{i_1,j_1\},\ldots , \{i_l,j_l\}\} \\ \text{pairing of }\{1,\ldots ,k\} }} \quad \prod_{m=1}^l \E[x_{i_m}x_{j_m}]
\end{equation}
In particular it follows that if $x_1,\ldots, x_p$ are independent standard Gaussian random variables, then
\[\E[x_1^{k_1}\ldots x_{p}^{k_p}]=\prod_{i=1}^p (2k_i)!! \quad.\]
\end{theorem}

\subsection{Unitary integration. Weingarten calculus}
\label{sec:Weingarten}

This section contains some basic material on unitary integration and Weingarten calculus. A more complete exposition of these matters can be found in \cite{collins2003moments, collins2006integration}. We start with the definition of the Weingarten function.

\begin{definition}
The unitary Weingarten function 
$\Wg(n,\sigma)$
is a function of a dimension parameter $n$ and of a permutation $\sigma$
in the symmetric group $\S_p$. 
It is the inverse of the function $\sigma \mapsto n^{\#  \sigma}$ under the convolution for the symmetric group ($\# \sigma$ denotes the number of cycles of the permutation $\sigma$).
\end{definition}

Note that the  function $\sigma \mapsto n^{\# \sigma}$ is invertible when $n$ is large, as it
behaves like $n^p\delta_e$
as $n\to\infty$.
If $n<p$ the function is not invertible any more.
For the definition to make sense, one needs
to take the  pseudo inverse
(we refer to \cite{collins2006integration} for historical references and further details). We  use the shorthand notation $\Wg(\sigma) = \Wg(n, \sigma)$ when the dimension parameter $n$ is clear from context.

The function $\Wg$  is used to compute integrals with respect to 
the Haar measure on the unitary group (we shall denote by $\U(n)$ the unitary group acting on an $n$-dimensional Hilbert space). The first theorem is as follows:

\begin{theorem}
\label{thm:Weingarten-formula}
 Let $n$ be a positive integer and
$\mathbf{i}=(i_1,\ldots ,i_p)$, $\mathbf{i'}=(i'_1,\ldots ,i'_p)$,
$\mathbf{j}=(j_1,\ldots ,j_p)$, $\mathbf{j'}=(j'_1,\ldots ,j'_p)$
be $p$-tuples of positive integers from $\{1, 2, \ldots, n\}$. Then
\begin{multline}
\label{eq:Weingarten} \int_{\U(n)}U_{i_1j_1} \cdots U_{i_pj_p}
\overline{U_{i'_1j'_1}} \cdots
\overline{U_{i'_pj'_p}}\ dU=\\
\sum_{\sigma, \tau\in \S_{p}}\delta_{i_1i'_{\sigma (1)}}\ldots
\delta_{i_p i'_{\sigma (p)}}\delta_{j_1j'_{\tau (1)}}\ldots
\delta_{j_p j'_{\tau (p)}} \Wg (n,\tau\sigma^{-1}).
\end{multline}

If $p\neq p'$ then
\begin{equation} \label{eq:Wg_diff} \int_{\U(n)}U_{i_{1}j_{1}} \cdots
U_{i_{p}j_{p}} \overline{U_{i'_{1}j'_{1}}} \cdots
\overline{U_{i'_{p'}j'_{p'}}}\ dU= 0.
\end{equation}
\end{theorem}

Since we  perform integration over large unitary groups, we are interested in the values of the Weingarten function in the limit $n \to \iy$. The following result encloses all the information 
we need for our computations
about the asymptotics of the $\Wg$ function; see \cite{collins2003moments} for a proof.

\begin{theorem}\label{thm:mob} For a permutation $\sigma \in \mathcal S_p$, let $\text{Cycles}(\sigma)$ denote the set of cycles of $\sigma$. Then
\begin{equation}
\Wg (n,\sigma )=(-1)^{n-\# \sigma}
\prod_{c\in \text{Cycles} (\sigma )}\Wg (n,c)(1+O(n^{-2}))
\end{equation}
and 
\begin{equation}
\Wg (n,(1,\ldots ,d) ) = (-1)^{d-1}c_{d-1}\prod_{-d+1\leq j \leq d-1}(n-j)^{-1}
\end{equation}
where $c_i=\frac{(2i)!}{(i+1)! \, i!}$ is the $i$-th Catalan number.
\end{theorem}
As a shorthand for the
quantities in Theorem \ref{thm:mob}, we introduce the function $\Mob$ on the symmetric
group. $\Mob$ is invariant under conjugation and multiplicative over the cycles; further, it satisfies
for any permutation $\sigma \in \S_p$:
\begin{equation}
\Wg(n,\sigma) = n^{-(p + |\sigma|)} (\Mob(\sigma) + O(n^{-2}))
\end{equation}
where $|\sigma |=p-\# \sigma $ is the \emph{length} of $\sigma$, i.e. the minimal number of transpositions that multiply to $\sigma$. We refer to \cite{collins2006integration} for details about the function $\Mob$.
We finish this section by a well known lemma which we will use several times towards the end of the paper. This result is contained in \cite{nica2006lectures}.
\begin{lemma}\label{lem:S_p}
The function
$d(\sigma,\tau) = |\sigma^{-1} \tau|$ is an integer valued distance on $\S_p$. Besides, it has the following properties:
\begin{itemize}
\item the diameter of $\S_p$ is $p-1$;
\item $d(\cdot, \cdot)$ is left and right translation invariant;
\item for three permutations $\sigma_1,\sigma_2, \tau \in \S_p$, the quantity $d(\tau,\sigma_1)+d(\tau,\sigma_2)$
has the same parity as $d(\sigma_1,\sigma_2)$;
\item the set of geodesic points between the identity permutation $\id$ and some permutation $\sigma \in \S_p$ is in bijection with the set of non-crossing partitions smaller than $\pi$, where the partition $\pi$ encodes the cycle structure of $\sigma$. Moreover, the preceding bijection preserves the lattice structure. 
\end{itemize}
\end{lemma}
Finally, we introduce a definition which generalizes the trace function: for some matrices $A_1, A_2, \ldots, A_p \in \M_n(\C)$ and some permutation $\sigma \in \S_p$, we define
\[\trace_\sigma(A_1, \ldots, A_p) = \prod_{\substack{c\in \text{Cycles} (\sigma) \\ c=(i_1 \, i_2 \, \cdots \, i_k)}} \trace \left( A_{i_1} A_{i_2} \cdots A_{i_k}\right).\]
We also put $\trace_\sigma(A) = \trace_\sigma(A, A, \ldots, A)$.

\subsection{Graphical interpretation of
Wick and Weingarten calculus}\label{sec:graphical-Wg-wick}

Our main motivation for the graphical 
calculus from to allow
to interpret nicely the above integration theorems
\ref{thm:wick-formula}, \ref{thm:Weingarten-formula}.
We consider first the case of the Weingarten calculus. 
The key to an interpretation
relies on the concept of \emph{removal} of boxes $U$ and $\ol U$.

A removal $r$ is a way to pair decorations of the $U$ and $\ol U$ boxes appearing in a diagram. 
It  consists in  a pairing $\alpha $ of the white decorations of $U$  boxes with the white decorations of $\ol U$ boxes, 
together with a pairing $\beta $ between the black decorations of $U$ boxes and the black decorations of $\ol U$ boxes. 
Assuming that $\D$ contains $p$ boxes of type $U$ and that the boxes $U$ (resp. $\ol U$) are labeled from $1$ to $p$, 
then $r=(\alpha,\beta)$ where $\alpha,\beta$ are permutations of $\mathcal{S}_p$. The set of all removals of $U$ and $\ol U$ boxes is denoted by $\Rem_U(\D)$.

A removal $r \in \Rem_U(\D)$, yields a new diagram $\D_r$ associated to $r$, which has the important property that it no longer contains boxes of type $U$ or $\ol U$. 
One starts by erasing the boxes $U$ and $\ol U$ but keeps the decorations attached to them. 
Assuming that one has labeled the erased boxes $U$ and $\ol U$ with integers from $\{1, \ldots, p\}$, one connects \emph{all} the (inner parts of the) \emph{white} decorations of the $i$-th erased $U$ box with the corresponding (inner parts of the) \emph{white} decorations of the $\alpha(i)$-th erased $\ol U$ box. In a similar manner, one uses the permutation $\beta$ to connect black decorations. 
In \cite{collins2010random}, we proved the following result:
\begin{theorem}\label{thm:Wg_diag}
The following holds true:
\[\E_U(\D)=\sum_{r=(\alpha, \beta) \in \Rem_U(\D)} \D_r \Wg (n, \alpha\beta^{-1}).\]
\end{theorem}
In the case where  diagrams also involve a  box $G$ corresponding to a \emph{Gaussian random matrix},
we are also able to compute the expected value
conditional to the $\sigma$-algebra of $G$
by graphical methods, yielding 
a new interpretation of Wick formula

Namely,
the expectation value of  a random diagram $\D$ can be computed by a \emph{removal} 
procedure as in the unitary case. Without loss of generality, we assume that we do not have in our diagram adjoints of Gaussian matrices, 
but instead their complex conjugate box. This assumption allows for a more straightforward use of the Wick Lemma \ref{thm:wick-formula}. 
As in the unitary case, we can assume that $\D$ contains only one type of random Gaussian box $G$; the other independent random Gaussian matrices are assumed constant at this stage as they shall be removed in the same manner afterwards. 

A removal of the diagram $\D$ is a pairing between \emph{Gaussian boxes} $G$ and their conjugates $\ol G$. The set of removals is denoted by $\Rem_G(\D)$ and it may be empty: if the number of $G$ boxes is different from the number of $\ol G$ boxes, then $\Rem_G(\D) = \emptyset$ (this is consistent with the first case of the Wick formula \eqref{eq:Wick}). Otherwise, a removal $r$ can identified with a permutation $\alpha \in \mathcal S_p$, where $p$ is the number of $G$ and $\ol G$ boxes. 
The main difference between the notion of a removal in the Gaussian and the Haar unitary cases is as follows: in the Haar unitary (Weingarten) case, a removal was associated with a \emph{pair of permutations}: 
one has to pair white decorations of $U$ and $\ol U$ boxes and, independently, black decorations of conjugate boxes. On the other hand, in the Gaussian/Wick case, one pairs {conjugate boxes}: white and black decorations are paired in an identical manner, hence only one permutation is needed to encode the removal.

To each removal $r$ associated to a permutation $\alpha \in \mathcal S_p$ corresponds a removed diagram $\D_r$ constructed as follows. One starts by erasing the boxes $G$ and $\ol G$, but keeps the decorations attached to these boxes. Then, the decorations (white \emph{and} black) of the $i$-th $G$ box are paired with the decorations of the $\alpha(i)$-th $\ol G$ box in a coherent manner, see Figure \ref{fig:expectation_Gaussian}.

\begin{figure}
\includegraphics[width=0.4\textwidth]{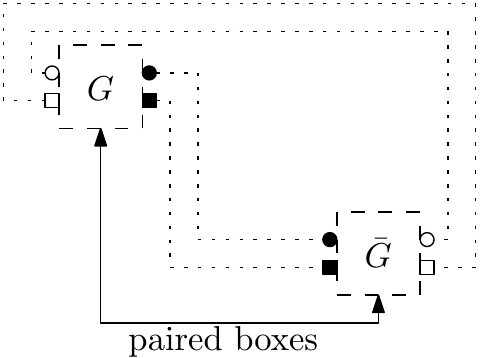}
\caption{Pairing of boxes in the Gaussian case}
\label{fig:expectation_Gaussian}
\end{figure}

The graphical reformulation of the Wick Lemma \ref{thm:wick-formula} becomes the following theorem, which we state without proof.

\begin{theorem}\label{thm:Wick_diag}
The following holds true:
\[\E_G[\D]=\sum_{r \in \Rem_G(\D)} \D_r .\]
\end{theorem}

\subsection{Some elements of free probability theory}
\label{sec:free-probability}

A {\it non-commutative probability space \it} is
an algebra $\mathcal A$ with unit endowed with a tracial 
state $\phi$. 
An element of $\mathcal A$ is called
a (non-commutative) random variable. In this paper we shall be mostly concerned with the non-commutative probability space of \emph{random matrices} $(\M_n(L^{\iy-}(\Omega, \P)), \E[n^{-1}\trace(\cdot)])$ (we use the standard notation $L^{\iy-}(\Omega, \P) = \cap_{p\geq 1} L^p(\Omega, \P)$). 

Let $\mathcal A_1, \ldots ,\mathcal A_k$ be subalgebras of $\mathcal A$ having the same unit as $\mathcal A$.
They are said to be \emph{free} if for all $a_i\in \mathcal  A_{j_i}$ ($i=1, \ldots, k$) 
such that $\phi(a_i)=0$, one has  
$$\phi(a_1\cdots a_k)=0$$
as soon as $j_1\neq j_2$, $j_2\neq j_3,\ldots ,j_{k-1}\neq j_k$.
Collections $S_{1},S_{2},\ldots $ of random variables are said to be 
free if the unital subalgebras they generate are free.

Let $(a_1,\ldots ,a_k)$ be a $k$-tuple of selfadjoint random variables and let
$\mathbb{C}\langle X_1 , \ldots , X_k \rangle$ be the
free $*$-algebra of non commutative polynomials on $\mathbb{C}$ generated by
the $k$ indeterminates $X_1, \ldots ,X_k$. 
The {\it joint distribution\it} of the family $\{a_i\}_{i=1}^k$ is the linear form
\begin{align*}
\mu_{(a_1,\ldots ,a_k)} : \C\langle X_1, \ldots ,X_k \rangle &\to \C \\
P &\mapsto \phi (P(a_1,\ldots ,a_k)).
\end{align*}
In the case of a single, self-adjoint random variable $x$, if the moments of $x$ coincide with those of a compactly supported probability measure $\mu$, i.e.
$$\forall p \geq 1, \qquad \phi(x^p) = \int t^p d\mu(t),$$
we say that $x$ has distribution $\mu$. The most important distribution in free probability theory is the semicircular distribution 
$$\mu_{SC(0,1)} = \frac{\sqrt{4-x^2}}{2\pi} \mathbf{1}_{[-2,2]}(x) dx,$$
which is, for reasons we will not get into, the free world equivalent of the Gaussian distribution in classical probability (see \cite[Lecture 8]{nica2006lectures} for the details). A random variable $x$ having distribution $\mu_{SC(0,1)}$ has the Catalan number for moments:
$$\phi(x^p) = \begin{cases}
\mathrm{Cat}_p := \frac{1}{p+1}\binom{2p}{p}, \qquad &\text{ if $p$ is even}\\
0, \qquad &\text{ if $p$ is odd.}
\end{cases}$$
More generally, if $x$ has distribution $\mu_{SC(0,1)}$, we say that $y=\sigma x + m$ has distribution 
\begin{equation}
\label{eq:def-semicircular}
\mu_{SC(m,\sigma^2)} = \frac{\sqrt{4\sigma^2-(x-m)^2}}{2\pi \sigma^2} \mathbf{1}_{[m=2\sigma,m+2\sigma]}(x) dx.
\end{equation}

\begin{figure}[htbp]
\begin{center}
\includegraphics[width=0.4\textwidth]{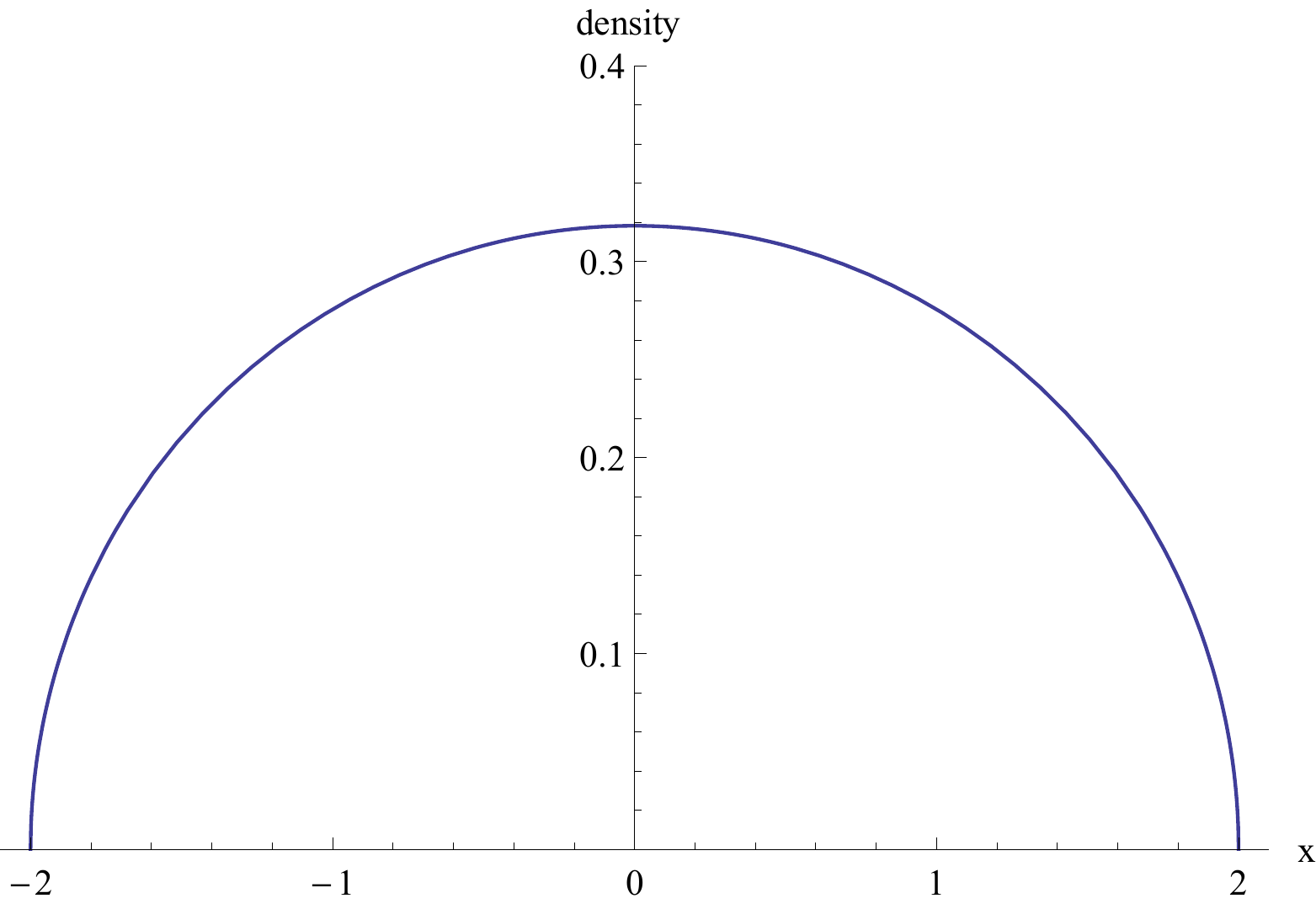} \quad \includegraphics[width=0.4\textwidth]{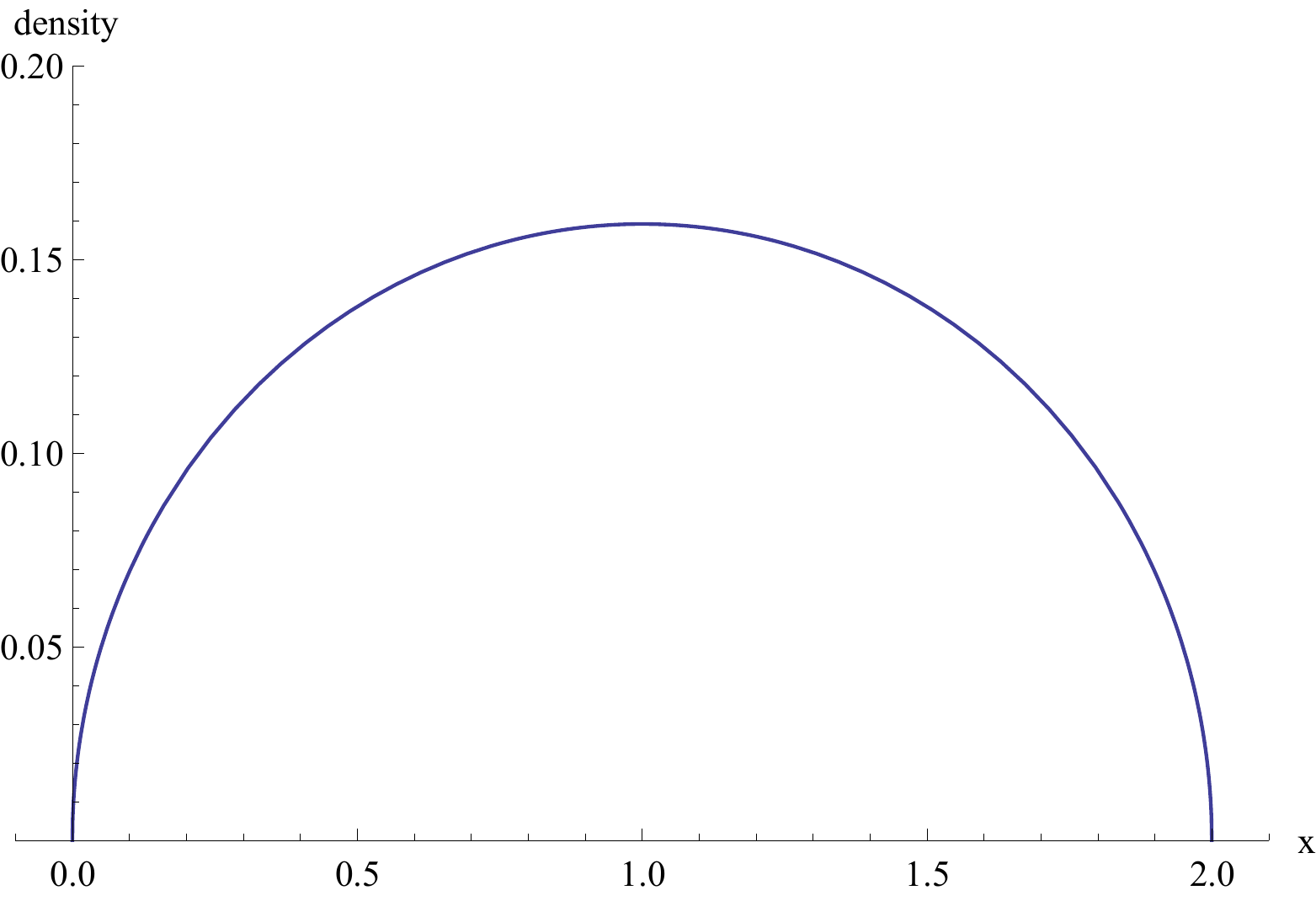}
\caption{The density of the semicircular distributions $\mu_{SC(0,1)}$ (left) and $\mu_{SC(1,1/4)}$ (right).}
\label{fig:semicircular}
\end{center}
\end{figure}

Given a $k$-tuple $(a_1,\ldots ,a_k)$ of free 
random variables such that the distribution of $a_i$ is $\mu_{a_i}$, the joint distribution
$\mu_{(a_1,\ldots ,a_k)}$ is uniquely determined by the
$\mu_{a_i}$'s.
A family $(a_1^{n},\ldots ,a_k^{n})_n$ of $k$-tuples of random
variables is said to \emph{converge in distribution} towards $(a_1,\ldots ,a_k)$
iff for all $P\in \C \langle X_1, \ldots ,X_k \rangle$, 
$\mu_{(a_1^n,\ldots ,a_k^n)}(P)$ converges towards
$\mu_{(a_1,\ldots ,a_k)}(P)$ as $n\to\infty$. 
Sequences of random variables  $(a_1^{n})_n,\ldots ,(a_k^{n})_n$ are called \emph{asymptotically free} as $n \to \iy$
iff the $k$-tuple $(a_1^{n},\ldots ,a_k^{n})_n$ converges in distribution towards a family of free random variables.

The following result was contained in \cite{voiculescu1998strengthened} (see also \cite{collins2006integration}).

\begin{theorem}\label{libre}
Let $\{U^{(n)}_k\}_{k \in \N}$ be a collection of independent 
Haar distributed random matrices of $\M_n (\C )$ and $\{W^{(n)}_k\}_{k\in \N}$ be a 
set of constant matrices of $\M_n (\C )$ 
admitting a joint limit distribution as $n \to \iy$ with respect to the
state $n^{-1}\trace$.
Then, almost surely,
 the family $\{U^{(n)}_k, W^{(n)}_k\}_{k \in \N}$ admits a limit $*$-distribution $\{u_k, w_k\}_{k \in \N}$ with respect to $n^{-1}\trace$, such that $u_1$, $u_2$, \ldots, $\{w_1, w_2, \ldots\}$ are free.
\end{theorem}

Given two free random variables $a,b\in\mathcal{A}$, the distribution $\mu_{a+b}$
is uniquely determined by $\mu_a$ and $\mu_b$. The free additive convolution
of $\mu_a$ and $\mu_b$ is defined by $\mu_a\boxplus\mu_b=\mu_{a+b}$. 
When $x=x^*\in\mathcal{A}$, we identify $\mu_x$ with
the spectral measure of $x$ with respect to $\tau$.
The operation $\boxplus$ induces a binary operation 
on the set of probability measures on $\mathbb{R}$.

\section{Entanglement of random quantum states}
\label{sec:entanglement}

\subsection{Probability distributions on the set of quantum states}
\label{sec:random-states}

\subsubsection{Random pure quantum states}

The first model for random quantum states we look at is the \emph{uniform measure on pure quantum states}. Indeed, the set of pure quantum states of a finite dimensional Hilbert space $\mathcal H = \mathbb C^d$ can be identified, up to a phase, with the set of points on the unit sphere of $\mathcal H$, $\{ x \in \mathbb C^d \, : \, \|x\| = 1\}$. On this set, there is a canonical probability measure, the uniform (or Lebesgue) measure.  

\begin{definition}\label{def:random-pure-state}
	A random pure quantum state $x \in \mathbb C^d$ is said to follow the \emph{uniform distribution} if $x$ is uniformly distributed on the unit sphere of $\mathbb C^d$. We denote the uniform distribution of pure states in $\mathbb C^d$ by $\chi_d$. 
\end{definition}

The uniform distribution has the following important properties \cite[Section 2.1]{nechita2007asymptotics}.
\begin{proposition}\label{prop:random-pure-states}
Let $x \in \mathbb C^d$ be a uniformly distributed pure quantum state, $x \sim \chi_d$. Then:
\begin{enumerate}
\item For any unitary operator $U \in \mathcal U_d$ ($U$ can either be fixed or random, but independent from $x$), the random pure state $Ux$ also has the uniform distribution, $Ux \sim \chi_d$.
\item If $G \in \mathbb C^d$ is random complex Gaussian vector, $X \sim \mathcal N_\mathbb C(0, I_n)$, then $X/\|X\|$ is a uniform quantum pure state, $X/\|X\| \sim \chi_d$.
\item Let $U$ be a random unitary matrix distributed along the Haar measure on $\mathcal U_n$  and let $y$ be the first column of $U$. Then $y \in \mathbb C^d$ is a uniform quantum pure state, $y \sim \chi_d$.
\end{enumerate}
\end{proposition}

In applications, whenever one needs to consider generic pure quantum states and that there is no underlying structure in the Hilbert space where the states live, the uniform measure is used indiscriminately. Later, in Section \ref{sec:random-states-graphs}, we shall encounter another probability distribution on a Hilbert space $\mathcal H$, which is to be used in the case where the space has a tensor product structure $\mathcal H = \mathcal H_1 \otimes \cdots \mathcal H_k$.

A different possibility was considered in \cite{nechita2013random}, starting from the first point in \ref{prop:random-pure-states}, and replacing the Haar unitary $U$ with the value of the unitary Brownian motion at some fixed time $t$ (recall that the Haar measure is recovered at the limit $t \to \infty$). The resulting measure depends on the time $t>0$ and on the initial vector $x$ on which the unitary acts. We refer the interested reader to \cite{nechita2013random} for the details.

\subsubsection{The induced ensemble}
\label{sec:induced-measure}

We introduce in this section a family of probability distributions on the set of (mixed) quantum states $\mathcal D_d$ which has a nice physical interpretation and, at the same time, a simple mathematical presentation. 

The following family was introduced by Braunstein in \cite{braunstein1996geometry} and studied by Hall \cite{hall1998random}, and later, in great detail, by  {\.Z}yczkowski and Sommers \cite{zyczkowski2001induced, sommers2004statistical}. 

\begin{definition}\label{def:induced-measures}
Given two positive integers $d,s$, consider a random pure quantum state $x \in \mathbb C^d \otimes \mathbb C^s$. The distribution of the random variable 
M$$\rho = [\mathrm{id}_d \otimes \mathrm{Tr}_s](xx^*) \in \mathcal D_d$$
is called the \emph{induced measure} of parameters $(d,s)$ and it is denoted by $\nu_{d,s}$.
\end{definition}

We gather in the following proposition some basic facts about the measures (for the proofs, see \cite{zyczkowski2001induced}).

\begin{proposition}
Let $\mathcal D_d \ni \rho \sim \nu_{d,s}$ be a density matrix having an induced distribution of parameters $(d,s)$.
\begin{enumerate}
\item With probability one, $\rho$ has rank $\min(d,s)$.
\item For any unitary operator $U \in \mathcal U_d$ (fixed or independent from $\rho$), the density matrix $U \rho U^*$ has the same distribution as $\rho$.
\item There exist a unitary matrix $U \in \mathcal U_d$ and a diagonal matrix $\Delta = \mathrm{diag}(\lambda_1, \ldots, \lambda_d)$ such that $U$ is Haar distributed, $U$ and $\Delta$ are independent, and $\rho = U \Delta U^*$; we say that the radial and the angular part of $\rho$ are independent.
\item The eigenvalues $(\lambda_1, \ldots, \lambda_d)$ have the following joint distribution:
$$C_{d,s} \mathbf{1}_{\lambda_1 + \cdots + \lambda_d = 1} \prod_{i=1}^d  \mathbf{1}_{\lambda_i \geq 0} \prod_{1 \leq i < j \leq d} (\lambda_i - \lambda_j)^2 \prod_{i=1}^d\lambda_i^{s-d},$$
where $C_{d,s}$ is the constant
$$C_{d,s} = \frac{\Gamma(ds)}{\prod_{i=0}^{d-1}\Gamma(s-i)\Gamma(d+1-i)}.$$
\end{enumerate}
\end{proposition}

\begin{remark}
\label{rem:induced-euclidean}
Importantly, in the case $s=d$, the distribution $\nu_{d,d}$ is precisely the \emph{Lebesgue measure} on the compact set $\mathcal D_d$, seen as a subset of the affine subspace $\{A \in \mathcal M_d^{sa}(\mathbb C) \, : \, \mathrm{Tr}(A) = 0\}$, see \cite[Section 2.4]{zyczkowski2001induced}. The measure $\nu_{d,d}$ is sometimes called the \emph{Hilbert-Schmidt measure}, since it is induced by the Euclidian, or Hilbert-Schmidt, distance. Note that the volume of $\mathcal D_d$ is given by \cite[Equation (4.5)]{zyczkowski2003hilbert}
$$\mathrm{vol}(\mathcal D_d) = \sqrt d (2\pi)^{d(d-1)/2} \frac{(d-1)!}{(d^2-1)!}.$$
\end{remark}

In \cite{nechita2007asymptotics}, the induced measures $\nu_{d,s}$ are shown to be closely related to the  Wishart ensemble $\mathcal W_{d,s}$ from Definition \ref{def:Wishart}.

\begin{proposition}\label{prop:induced-measure-properties}
Let $W \in \mathcal M_d(\mathbb C)$ be a Wishart matrix of parameters $(d,s)$ and put $\rho:=W/ \mathrm{Tr}(W) \in \mathcal D_d$. Then
\begin{enumerate}
\item The random variables $\rho$ and $\mathrm{Tr}(W)$ are independent. 
\item The distribution of $\mathrm{Tr}(W)$ is chi-squared, with $ds$ degrees of freedom.
\item The random density matrix $\rho$ follows the induced measure of parameters $(d,s)$, i.e.~ $\rho \sim \nu_{d,s}$.
\item The random variable $W$, conditioned on the (zero probability) event $\mathrm{Tr}(W)=1$, has distribution $\nu_{d,s}$. 
\end{enumerate}
\end{proposition}

Let us now discuss the asymptotic behavior of the probability measures $\nu_{d,s}$. We first consider the ``trivial'' regime, where $d$ is fixed and $s \to \infty$. The result here is as follows, see \cite{nechita2007asymptotics}.

\begin{proposition}
For a fixed dimension $d$, consider a sequence of random density matrices $(\rho_s)_s$ having distribution $\rho_s \sim \nu_{d,s}$. Then, almost surely as $s \to \infty$, $\rho_s \to d^{-1}I_d$.
\end{proposition}

The interesting scaling is the fixed ration one, where both $d$ and $s = s_d$ grow to infinity, in such a way that $s_d/d \to c$, for a fixed constant $c \in (0,\infty)$, The next result is an easy consequence of Theorem \ref{thm:marchenko-pastur} and Proposition \ref{prop:induced-measure-properties}. 

\begin{proposition}
For a fixed positive constant $c$, consider a sequence of random density matrices $(\rho_d)_d$ having distribution $\rho_d \sim \nu_{d,s_d}$; here we assume that $s_d \sim cd$ as $d \to \infty$. Then, almost surely as $d \to \infty$, the empirical eigenvalue distribution of the random matrix $s_d \rho_d$ converges weakly to the Mar{\v{c}}enko-Pastur distribution $\pi_c$ from \eqref{eq:Marchenko-Pastur}
$$\lim_{d \to \infty} \frac{1}{d} \sum_{i=1}^d \delta_{s_d \lambda_i(\rho_d)} = \pi_c.$$
\end{proposition}

Informally, the result above can be stated as follows: consider a tensor product Hilbert space $\mathcal H = \mathbb C^d \otimes \mathbb C^{\lfloor cd \rfloor}$ and random, uniform pure state $\psi \in \mathcal H$. Then, the eigenvalues of the partial trace $\rho = [\mathrm{id} \otimes \mathrm{Tr}](\psi \psi^*)$ are, up to a scaling of $cd$, distributed along the Mar{\v{c}}enko-Pastur distribution $\pi_c$  \eqref{eq:Marchenko-Pastur}.

Finally, as suggested by Proposition \ref{prop:induced-measure-properties}, in order to simulate on a computer quantum states having distribution $\nu_{d,s}$, one sets
$$\rho = \frac{GG^*}{\mathrm{Tr}(GG^*)},$$
where $G \in \mathcal M_{d \times s}(\mathbb C)$ is an element from the \emph{Ginibre ensemble}, i.e.~$G$ has i.i.d.~standard complex Gaussian entries; see \cite[Section III.D]{zyczkowski2011generating}.

\subsubsection{The Bures measure}

The Bures metric on the set of density matrices (see \cite{bengtsson2006geometry}) is defined as 
$$d_B(\rho,\sigma) = \sqrt{2 - 2\mathrm{Tr}[(\sqrt \rho \sigma \sqrt \rho)^{1/2}]}.$$
From this metric, one can define a probability distribution $\nu_B$ on $\mathcal D_d$, by asking that Bures balls of equal radius have the same volume. 

The properties of the measure $\nu_B$ have been extensively studied in \cite{hall1998random,sommers2003bures,osipov2010random}, we recall in the next proposition the main facts. 

\begin{proposition}
Let $\rho \in \mathcal D_d$ be a random density matrix having distribution $\nu_B$. Then
\begin{enumerate}
\item The eigenvalues $\lambda_1, \ldots, \lambda_d$ of $\rho$ have distribution
$$C_{B} \mathbf{1}_{\lambda_1 + \cdots + \lambda_d = 1} \prod_{i=1}^d  \mathbf{1}_{\lambda_i > 0} \lambda_i^{-1/2} \prod_{1 \leq i < j \leq d} \frac{(\lambda_i - \lambda_j)^2}{\lambda_i + \lambda_j},$$
where the constant $C_B$ reads
$$C_B  = 2^{d^2-d} \frac{\Gamma(d^2/2)}{\pi^{d/2} \prod_{i=1}^d \Gamma(i+1)}.$$
\item If $A \in \mathcal M_d(\mathbb C)$ is a random Ginibre matrix and $U \in \mathcal U_d$ is a Haar random unitary independent from $A$, then the random matrix 
$$\sigma = \frac{(I+U)AA^*(I+U)^*}{\mathrm{Tr}[(I+U)AA^*(I+U)^*]}$$
has distribution $\nu_B$.
\end{enumerate}
\end{proposition}

\subsubsection{Random states associated to graphs}
\label{sec:random-states-graphs}

The probability distributions on $\mathcal D_d$ we have considered so far do not make any assumptions on the internal structure of the underlying Hilbert space $\mathbb C^d$. To address this issue, in \cite{collins2010randoma,collins2013area} the authors introduce and study a new family of ensembles of density matrices, called \emph{random graph states}, which encode the underlying structure of the Hilbert space. We introduce next these distributions, referring the interested reader to \cite{collins2010randoma,collins2013area} for the details. 

Consider a graph $G=(V,E)$ having $k$ vertices $V_1, \ldots, V_k$ and $m$ edges $E_1, \ldots, E_m$. Let $N$ be a fixed positive integer, and consider the (total) Hilbert space 
$$\mathcal H = \bigotimes_{i=1}^k \mathcal H_i,$$
where $\mathcal H_i = (\mathbb C^N)^{\otimes d_i}$ is the \emph{local Hilbert space} at vertex $i$ and $d_i$ is the degree of $V_i$ in $G$. Each copy of $\mathbb C^N$ inside $\mathcal H_i$ is associated to some edge $E_j$ incident to $V_i$, in such a way that the total Hilbert space admits two decompositions, relative to vertices and edges:
$$\mathcal H = \bigotimes_{i=1}^k \mathcal H_i = \bigotimes_{j=1}^m \mathcal K_j \simeq (\mathbb C^N)^{\otimes 2m},$$
where $\mathcal K_j = \mathbb C^N \otimes C^N$. Define now the following random pure state
$$\varphi_G = \left[ \bigotimes_{i=1}^k U_i \right] \left[ \bigotimes_{j=1}^k \omega_j \right],$$
where $\{U_i\}_{i=1}^k$ are i.i.d.~Haar distributed random unitary matrices acting on the local Hilbert spaces at the vertices, and $\omega_j$ are maximally entangled states \eqref{eq:maximally-entangled-state}. Note that in the above expression, the unitary operators ``mix'' the product of maximally entangled states at the vertices, yielding, in general, a global entangled state. 

Let us now define mixed quantum states with the above formalism. For a subset $S \subseteq \{1, 2, \ldots, 2m\}$ of copies of $\mathbb C^N$, define
$$\rho_{G,S} = [\mathrm{id}_S \otimes \mathrm{Tr}_{S^c}](\varphi_G \varphi_G^*) \in \mathcal D_{N^{|S|}}.$$
The statistical properties of the distribution of $\rho_{G,S}$ are studied in \cite[Section 5]{collins2010randoma}.

Here, we show that the area law holds \emph{exactly} for graph states, 
provided that the marginal under consideration satisfies a
 particular condition, called \emph{adaptability}.

To any graph state we associate two partitions of the set of $n=2m$ subspaces:
 a vertex partition $\mathcal P_\text{vertex}$ 
which encodes the vertices of the graph, and a pair partition $\mathcal P_\text{edge}$
 which encodes the edges (corresponding to maximally entangled states). More precisely, 
 two subsystems $\H_i$ and $\H_j$ belong to the same block of $\mathcal P_\text{vertex}$ if they are attached to the 
 same vertex of the initial graph. Each edge $(i,j)$ of the graph contributes a block of size two $\{i,j\}$ to the edge 
 partition $\mathcal P_\text{edge}$.
 Recall that a marginal of a random graph state $\varphi_G \varphi_G^*$
 is specified by a  2-set partition $\mathcal P_\text{trace} = \{S, T\}$. 
 
Let us introduce now a fundamental property of the (random) quantum states associated to graphs. 

\begin{definition}
A marginal $\rho_S$ is called \emph{adapted} if
\begin{equation}
\mathcal P_\text{trace} \geq \mathcal P_\text{vertex}
\end{equation}
for the usual refinement order on partitions. In other words, 
a marginal is adapted if and only if the number of traced out systems in 
each vertex is either zero or maximal. If this is the case, then 
the partition boundary,  which splits the graph into parts $ \{S, T\}$,
does not cross any vertices of the graph.
\end{definition}

Because of the above property, for adapted marginals, we can speak about \emph{traced out vertices}, because if one subsystem of a vertex is traced out, then all
 the other systems of that vertex are also traced out. 
We now define precisely what we mean by \emph{area laws} in the context of quantum states associated to graphs. The partition $\{S, T \}$ defines a boundary between 
the set of vertices that are traced out and vertices that survive.

\begin{definition}\label{def:boundary-volume-adapted}
The \emph{boundary} of the adapted partition $\{S, T\}$ is defined as the set of all (unoriented) edges $e=\{i_S, j_T\}$ in the graph
 state with the property that $i_S \in S$ and $j_T \in T$. Equivalently, it is the set of edges of the type
 $\includegraphics{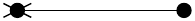}$. The boundary of a partition  shall be denoted by $\partial S$.

The \emph{area} of this boundary is its cardinality $|\partial S |$,  i.e. the number of edges between $S$ and $T$.
\end{definition}

It was shown in \cite{collins2013area} that the \emph{area law} holds \emph{exactly} for adapted marginals of graph states, where we allow arbitrary dimensions of subsystem. Note that, for a given (boundary) edge $\{i,j\}$, we have $d_i = d_j$, the common dimension of the maximally entangled state corresponding to the edge $\{i,j\}$. The following result follows from linear algebra considerations, and one does not need random Haar unitary operators in this case. 

\begin{proposition}\label{prop:area-law-adapted}
Let $\rho_S$ be an \emph{adapted} marginal of a graph state $\varphi_G$. Then, the entropy of $\rho_S$ has the following 
\emph{exact, deterministic} value:

\begin{equation}
H(\rho_S) = |\partial S| \log N.
\end{equation}
\end{proposition}

\begin{figure}[htbp]
\centering
\includegraphics{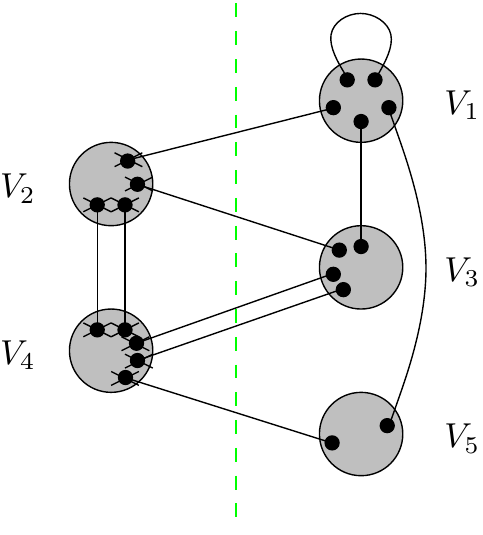}
\caption{An adapted marginal for a graph state. The dashed (green) line represents the \emph{boundary} between the traced--out subsystems $T$ and the surviving subsystems $S$.}
\label{fig:compatible_ex}
\end{figure}
For the system corresponding to the graph shown in Figure \ref{fig:compatible_ex} 
with all subsystems of size $N$  the von Neumann entropy reads
\begin{equation}
H(\rho_S) = 5 \log N .
\end{equation}
This follows from the fact that $\rho_S$ is in this case a unitary conjugation of a maximally mixed 
state of size $N^5$ with an arbitrary pure state of size $N^6$.

We refer the reader to Section \ref{sec:area-laws} for a more general result in this direction (for non-adapted marginals).

\subsection{Moments. Average entropy}

In this section we present results concerning certain quantities of interest in quantum information theory, and in particular their average values over the different ensembles introduced previously. 

Let us start with the case of the uniform measure on the set of pure quantum states. The statistics of the coordinates of a uniform random pure state can be obtained by the so-called \emph{spherical integrals} \cite[Section 2.7]{folland2013real}. The following result could also be deduced from the Wick formula in Section \ref{sec:Gaussian-random-variables} or from the Weingarten formula in Section \ref{sec:Weingarten}.

\begin{lemma}
For any non-negative integers $\alpha_1,\ldots, \alpha_d \geq 0$, we have
$$\mathbb E_{x \sim \chi_d}\left[ |x_1|^{2\alpha_1} |x_2|^{2\alpha_2} \cdots |x_d|^{2\alpha_d} \right] = (d-1)!\frac{\alpha_1! \alpha_2! \cdots \alpha_d!}{(d-1 + \alpha_1 + \alpha_2 + \cdots + \alpha_d)!}.$$
\end{lemma}

We move now to the case of random density matrices having the induced distributions $\nu_{d,s}$ discussed in Section \ref{sec:induced-measure}. Using the relation between this distribution and the Wishart ensemble, the following result has been shown in \cite{sommers2004statistical,nechita2007asymptotics}. 

\begin{proposition}
The moments of a random density matrix $\rho \in \mathcal D_d$ having distribution $\nu_{d,s}$ are given by 
$$\mathbb E \mathrm{Tr}(\rho^q) =
\frac{\Gamma(ds)}{\Gamma(ds+q)}\sum_{j=1}^{q}{(-1)^{j-1}
\frac{[s+q-j]_q [d+q-j]_q }{(q-j)! (j-1)!}},$$
where $[a]_q = a(a-1)\cdots(a-q+1)$In particular, the first few moments read
\begin{align*}
&\mathbb E \mathrm{Tr}(\rho^2) = \frac{d+s}{ds+1}\\
&\mathbb E \mathrm{Tr}(\rho^3) = \frac{d^2+3ds+s^2+1}{(ds+1)(ds+2)}\\
&\mathbb E \mathrm{Tr}(\rho^4) = \frac{d^3+6d^2s+6ds^2+s^3+5d+5s}{(ds+1)(ds+2)(ds+3)}.
\end{align*}
\end{proposition}

The average entropy of a random density matrix was conjectured by Page in \cite{page1993average} and later proved in \cite{foong1994proof,sanchez-ruiz1995simple,sen1996average}. 

\begin{proposition}
The average von Neumann entropy of a random density matrix having distribution $\nu_{d,s}$ is 
$$\mathbb E H(\rho) = \sum_{i = s+1}^{ds}{\frac 1 i} - \frac{d-1}{2s}.$$
\end{proposition}

\subsection{Entanglement}

The notion of quantum entanglement has been recognized to be at the center of quantum mechanics from the early days of the theory. The reader interested in entanglement theory is referred to the excellent review paper \cite{horodecki2009quantum}. In this work, we will only deal with \emph{bipartite entanglement}, which is defined as follows. First, we say that a quantum state $\rho \in \mathcal D_{nk}$ is \emph{separable} iff it can be written as a convex combination of tensor product states:
$$\rho = \sum_{i=1}^r p_i \sigma_i \otimes \tau_i,$$
where $\sigma_i \in \mathcal D_n$, $\tau_i \in \mathcal D_k$ and $(p_i)$ is a probability vector: $p_i \geq 0$ and $\sum_i p_i =1$. The set of separable states is denoted by $\mathcal{SEP}_{n,k} \subseteq \mathcal D_{nk}$ and the states in its complement are called entangled. 

In this section, we are going to review some results about the (Euclidean) volume of the set of separable states. Equivalently, volumes can be expressed, up to a factor, from the probability that a quantum state is separable, under the induced measure $\nu_{nk,nk}$, see Remark \ref{rem:induced-euclidean}.

The first result in this direction is quite remarkable \cite{gurvits2002largest}. It has many interesting corollaries, one of them being that the set $\mathcal{SEP}$ of separable states has non-empty interior.

\begin{proposition}\label{prop:separable-ball}
The largest Euclidean ball centered at the maximally mixed state $I/(nk)$ and contained in $\mathcal D_{nk}$ is separable and has radius $[nk(nk-1)]^{-1/2}$.
\end{proposition}

In the case of the Euclidean measure $\nu_{nk,nk}$ is has been shown in \cite[Theorem 1]{aubrun2006tensor} that the ratio between the volume of $\mathcal{SEP}_{n,n}$ and $\mathcal D_{n^2}$ vanishes when $n \to \infty$. In the case where the parameter $s$ of the induced measure $\nu_{nk,s}$ grows to infinity, while $n$ and $k$ are kept fixed, the measure $\nu_{nk,s}$ concentrates around the maximally mixed state $I_{nk}$ (see Proposition \ref{prop:induced-measure-properties}), so 
$$\lim_{s \to \infty} \mathbb P_{\nu_{nk,s}}[ \rho \in \mathcal{SEP}_{nk}] = 1.$$

More precise estimates have been obtained in \cite{aubrun2014entanglement} in the case of the induced measures. In order to present these results, we need first to introduce the concept of \emph{thresholds}. 

Consider a family of sets of density matrices $X_d \subseteq \mathcal D_d$. The idea of a threshold captures the behavior of the probability that a quantum state $\rho \in \mathcal D_d$ is an element of $X_d$, when the probability is measured with the induced measure $\nu_{d,s}$; we would like to know, when $d \to \infty$, for which values of the parameter $s$, the probability vanishes or becomes close to $1$. More precisely, we say that a \emph{threshold phenomenon} with value $c_0$ on the scale $f$ occurs when the following holds: let $s_d \sim c f(d)$ for a constant $c>0$; Then
\begin{enumerate}
\item If $c < c_0$, $\lim_{d \to \infty} \mathbb P_{\nu_{d,s_d}}[ \rho \in X_d] = 0$.
\item If $c > c_0$, $\lim_{d \to \infty} \mathbb P_{\nu_{d,s_d}}[ \rho \in X_d] = 1$.
\end{enumerate}
This definition was first considered in the Quantum Information Theory literature by Aubrun in \cite{aubrun2012partial} to study the PPT criterion (see next section). 

We state now the main result in \cite{aubrun2014entanglement}, regarding the threshold for the sets $\mathcal{SEP}_{n,n}$. The following statement corresponds to \cite[Theorem 2.3]{aubrun2014entanglement}, which deals with the so-called \emph{balanced regime} $k=n$. For the \emph{unbalanced regime} $k \neq n$, see \cite[Section 7.2]{aubrun2014entanglement}. 

\begin{theorem}\label{thm:threshold-sep}
There exist constants $c,C$ and a function $f(n)$ satisfying
$$cn^3 < f(n) < Cn^3 \log^2(n)$$
such that
\begin{enumerate}
\item If $s_n < f(n)$, $\lim_{n \to \infty} \mathbb P_{\nu_{n^2,s_n}}[ \rho \in \mathcal{SEP}_{n,n}] = 0$.
\item If $s_n > f(n)$, $\lim_{n \to \infty} \mathbb P_{\nu_{n^2,s_n}}[ \rho \in \mathcal{SEP}_{n,n}] = 1$.
\end{enumerate}
\end{theorem}

Note that the above result does not enter precisely in the threshold framework, as it was defined just above; one would need to eliminate the logarithm factors and to compute exactly the constants in the statement above to achieve this, see Question \ref{qst:threshold-sep}. The result is nevertheless an important achievement, given the fact that questions dealing directly with the set of separable states are usually very difficult.

\subsection{Entanglement criteria}
\label{sec:thresholds}

The question whether a given mixed quantum state is separable or entangled has been proven to be an NP-hard one \cite{gurvits2003classical}. To circumvent this worse-case intractability, \emph{entanglement criteria} are used. These are efficiently computable conditions which are necessary for separability; in other words, an entanglement criterion is a (usually convex) super-set $\mathcal X_d$ of the set of separable states, for which the membership problem is efficiently solvable. As in the previous section, from a probabilistic point of view, estimating the probability that a random quantum state (sampled from the induced ensemble) is an element of $\mathcal X_d$ is central. In what follows we shall tackle this problem for different entanglement criteria in the framework of thresholds. 

Let us start with the most used example, the \emph{positive partial transpose} criterion (PPT). The PPT criterion has been introduced by Peres in \cite{peres1996separability}: if a quantum state $\rho \in \mathcal D_{nk}$ is separable, then
$$\rho^\Gamma := [\mathrm{id} \otimes \mathrm{transp}](\rho) \geq 0.$$
Note that the positivity of $\rho^\Gamma$ is equivalent to the positivity of $\rho^{\scriptsize{\reflectbox{$\Gamma$}}} = [\mathrm{transp} \otimes \mathrm{id}](\rho)$,
so it does not matter on which tensor factor the transpose application acts. We denote by $\mathcal{PPT}_{n,k}$ the set of PPT states
$$\mathcal{PPT}_{n,k}:=\{\rho \in \mathcal{D}_{nk} \, : \, \rho^\Gamma \geq 0\} \supseteq \mathcal{SEP}_{n,k}.$$
This necessary condition for separability has been shown to be also sufficient for qubit-qubit and qubit-qutrit systems ($nk \leq 6$) in \cite{horodecki1996separability}. The PPT criterion for random quantum states has first been studied numerically in \cite{znidaric2007detecting}. The analytic results in the following proposition are from \cite{aubrun2012partial} (in the balanced case) and from \cite{banica2013asymptotic} (in the unbalanced case); see also \cite{fukuda2013partial} for some improvements in the balanced case and the relation to meanders. 

\begin{proposition}
\label{prop:thresholds-PPT}
Consider a sequence $\rho_n \in \mathcal D_{nk_n}$ of random quantum states from the induced ensemble $\nu_{nk_n,cnk_n}$, where $k_n$ is a function of $n$ and $c$ is a positive constant.

In the balanced regime $k_n = n$, the (properly rescaled) empirical eigenvalue distribution of the states $\rho_n$ converges to a semicircular measure $\mu_{SC(1,1/c)}$ of mean $1$ and variance $1/c$, see \eqref{eq:def-semicircular}. In particular, the threshold for the sets $\mathcal{PPT}_{n,n}$ ($n \to \infty$) is $c_0 = 4$.

In the unbalanced regime $k_n=k$ fixed, the (properly rescaled) empirical eigenvalue distribution of the states $\rho_n$ converges to a free difference of free Poisson distributions (see Section \ref{sec:free-probability} for the definitions)
$$\pi_{ck(k+1)/2} \boxminus \pi_{ck(k-1)/2}.$$
In particular, the threshold for the sets $\mathcal{PPT}_{n,k}$ ($k$ fixed, $n \to \infty$) is 
$$c_0 = 2+2\sqrt{1-\frac{1}{k^2}}.$$
\end{proposition}

We consider next the \emph{reduction} criterion (RED). Introduced in \cite{horodecki1999reduction,cerf1999reduction}, the reduction criterion states that if a bipartite quantum state $\rho \in \mathcal D_{nk}$ is separable, then
$$\rho^{red}:= [\mathrm{id} \otimes R](\rho) \geq 0,$$
where $R:\mathcal M_k(\mathbb C) \to \mathcal M_k(\mathbb C)$ is the \emph{reduction map}, 
$$R(X) = I_k\cdot \mathrm{Tr}(X) - X.$$
We denote by $\mathcal{RED}_{n,k}$ the set of quantum states having positive-semidefinite reductions (on the second subsystem)
$$\mathcal{RED}_{n,k}:=\{\rho \in \mathcal{D}_{nk} \, : \, \rho^{red} \geq 0\} \supseteq \mathcal{SEP}_{n,k}.$$
Several remarks are in order at this point. First, it is worth mentioning that in the literature, the reduction criterion is sometimes defined to ask that \emph{both} reductions, on the first and on the second subsystems, are positive-semidefinite; since going from one reduction to the other one can be done by simply swapping the roles of $\mathbb C^n$ and $\mathbb C^k$, we focus in this work on the reduction on the second subsystem. We gather in the next lemma some basic properties of the set $\mathcal{RED}_{n,k}$, see, e.g.~\cite{jivulescu2014reduction} for the proof. 

\begin{lemma}
The reduction criterion is, in general, weaker than the PPT criterion:
$$\mathcal{SEP}_{n,k} \subseteq \mathcal{PPT}_{n,k} \subseteq \mathcal{RED}_{n,k} \subseteq \mathcal{D}_{nk}.$$
However, at $k=2$ (i.e.~when the system on which the reduction map acts is a qubit), the two criteria are equivalent
$$\mathcal{RED}_{n,k} = \mathcal{PPT}_{n,k}.$$
\end{lemma}

Although the reduction criterion is weaker than the PPT criterion for the purpose of detecting entanglement, its interest stems from the connection with the distillability of quantum states, see \cite{horodecki1998mixed}.

We gather in the next proposition the values of the thresholds for the sets $\mathcal{RED}_{n,k}$. Since, in the case of the reduction criterion, the tensor factor on which the reduction map $R$ acts does matter, we need to consider two unbalanced regimes: one where $n$ is fixed and $k \to \infty$, and a second one where $n \to \infty$ and $k$ is kept fixed. The results below have been obtained in \cite{jivulescu2014reduction} (for the second unbalanced regime) and in \cite{jivulescu2015thresholds} (for the balanced regime and the first unbalanced regime).

\begin{proposition}
The thresholds for the sets $\mathcal{RED}_{n,k}$ are as follows:
\begin{enumerate}
\item In the balanced regime, where both $n,k \to \infty$, the threshold value for the parameter $s$ of the induced measure $\nu_{nk,s}$ is on the scale $s \sim cn$ at the value $c_0=1$.
\item In the first unbalanced regime, where $n$ is fixed and $k \to \infty$, the threshold value for the parameter $s$ of the induced measure $\nu_{nk,s}$ is on the scale $s \sim c$ at the value $c_0=n$.
\item In the second unbalanced regime, where $k$ is fixed and $n \to \infty$, the threshold value for the parameter $s$ of the induced measure $\nu_{nk,s}$ is on the scale $s \sim cnk$ at the value 
$$c_0=\frac{(1+\sqrt{k+1})^2}{k(k-1)}.$$
\end{enumerate}
\end{proposition}

Let us mention now that both thresholds for the PPT and the RED criterion, in the unbalanced case, have been treated, in a unified manner, in the recent preprint \cite{arizmendi2015asymptotic}. A general framework is developed in \cite{arizmendi2015asymptotic} in which many examples of entanglement criteria fit. 

Criteria of the type $[\mathrm{id} \otimes f](\rho) \geq 0$ have been studied from a random matrix theory perspective in \cite{collins2015random} in the case of random linear maps $f$. In \cite{collins2015random}, the authors introduce a family of entanglement criteria index by probability measures. The main idea is to consider maps $f$ between matrix algebras obtained from random Choi matrices. More precisely, consider a compactly supported probability measure $\mu$, and let $X_d \in \mathcal M_{nd}(\mathbb C)$ a sequence of unitarily invariant random matrices converging in distribution, as $d \to \infty$, to $\mu$ ($n$ being kept fixed). Let $f_d: \mathcal M_n(\mathbb C) \to \mathcal M_d(\mathbb C)$ be a (random) linear map such that the Choi matrix \eqref{eq:def-Choi-matrix} of $f_d$ is $X_d$. Then, the positivity of the map $f_d$, asymptotically as $d \to \infty$, depends only on $\mu$ and its free additive convolution powers \cite[Theorem 4.2]{collins2015random} (see Section \ref{sec:free-probability} for the definition of convolutions in free probability). 

\begin{theorem}\label{thm:k-positivity-from-mu}
The sequence of random linear maps $(f_{d})_d$ has the following properties:
\begin{enumerate}
\item If $\mathrm{supp}(\mu^{\boxplus n/k})\subset (0,\infty)$, then, almost surely as $d \to \infty$, $f_d$ is $k$-positive.
\item If $\mathrm{supp}(\mu^{\boxplus n/k})\cap (-\infty ,0)\neq \emptyset$, then, almost surely as $d \to \infty$, $f_d$ is not $k$-positive.
\end{enumerate}
\end{theorem}

From the above result, if follows that probability measures $\mu$ with the property that the maps they yield are positive, but not completely positive, give interesting entanglement criteria. It was shown in \cite[Theorem 5.4]{collins2015random} that such maps can be obtained from shifted semicircular measures \eqref{eq:def-semicircular}, and that they can detect PPT entanglement. The global usefulness of such entanglement criteria is left open (see Question \ref{qst:random-entanglement-criteria}).

We discuss next the \emph{realignment} criterion (RLN), introduced in \cite{rudolph2003cross,chen2002matrix}, is of different nature than the two other criteria we already discussed. For any matrix $X \in \mathcal M_n(\mathbb C) \otimes \mathcal M_k(\mathbb C)$, define
$$X^{rln} = L(X) \in \mathcal M_{n^2 \times k^2}(\mathbb C),$$
where $L$ is the realignment map, defined on elementary tensors by
$$L(e_ie_j^* \otimes f_a f_b^*) = e_if_a^* \otimes e_jf_b^*.$$
The realignment criterion states that a separable quantum state $\rho \in \mathcal D_{nk}$ satisfies
$$\|\rho^{rln}\|_1 \leq 1,$$
where $\|\cdot \|_1$ is the Schatten 1-norm (or the nuclear norm). As usual, we denote by $\mathcal{RLN}_{n,k}$ the set of quantum states satisfying the realignment criterion
$$\mathcal{RLN}_{n,k}:=\{\rho \in \mathcal{D}_{nk} \, : \, \|\rho^{rln}\|_1 \leq 1\} \supseteq \mathcal{SEP}_{n,k}.$$

The realignment criterion is not comparable to the PPT criterion, hence there are PPT entangled states detected by the RLN criterion. Since the inclusion partial relation cannot be used to compare the two sets/criteria, the notion of threshold is particularly interesting in this situation. The result below is from \cite{aubrun2012realigning}.

\begin{proposition}
The thresholds for the sets $\mathcal{RLN}_{n,k}$ are as follows:
\begin{enumerate}
\item In the balanced regime, where $n=k \to \infty$, the threshold value for the parameter $s$ of the induced measure $\nu_{n^2,s}$ is on the scale $s \sim cn^2$ at the value $c_0=(8/3\pi)^2 \simeq 0.72$.
\item In the unbalanced regime, where $n \to \infty$ and $k$ is fixed, the threshold value for the parameter $s$ of the induced measure $\nu_{nk,s}$ is on the scale $s \sim c$ at the value $c_0=k^2$.
\end{enumerate}
In particular, comparing the values above with the ones in Proposition \ref{prop:thresholds-PPT}, one can conclude that, from a volume perspective, the realignment criterion is weaker than the PPT criterion (i.e.~the thresholds for RLN are smaller than the thresholds for PPT).
\end{proposition}

We gather in Table \ref{tbl:thresholds-criteria} the values of the thresholds for the different entanglement criteria discussed in this section, as well as for the set of separable states itself. The striking feature of these values is the fact that the (bounds for the) thresholds for the set $\mathcal{SEP}$, obtained in the important work \cite{aubrun2014entanglement}, are one order of magnitude above the thresholds for the various entanglement criteria. This means that, from a volume perspective, the set $\mathcal{SEP}$ is much smaller than the set of states satisfying the different entanglement criteria. 

\begin{table}
\begin{center}
\begin{tabular}{cc|c|c|c|}
\cline{2-4}
&\multicolumn{1}{|c|}{\multirow{1}{*}{Balanced regime}} &\multicolumn{2}{|c|}{Unbalanced regime}\\
&\multicolumn{1}{|c|}{\multirow{1}{*}{$n,k\to\infty$}}
&\multicolumn{2}{|c|}{\multirow{1}{*}{$n \to\infty$, $k$ fixed}}\\
\cline{1-4}
\multicolumn{1}{ |c  }{$\mathcal{SEP}$}
&\multicolumn{1}{ |c| }{\quad $n^3\lesssim s\lesssim n^3\log^2 n$ \, $[n=k]$ \quad}
&\multicolumn{2}{ |c| }{\quad$nk^2\lesssim s\lesssim nk^2\log^2 (nk)$ \quad} \\
\cline{1-4}
\multicolumn{1}{|c}{\multirow{2}{*}{$\mathcal{PPT}$ } }
&\multicolumn{1}{|c|}{$s\sim cnk$}
&\multicolumn{2}{|c|}{$s\sim cnk$}  \\
\cline{2-4}
\multicolumn{1}{ |c  }{}
&\multicolumn{1}{ |c| }{$c_0=4$ \, $[n=k]$}
&\multicolumn{2}{c|}{$c_0=2+2\sqrt{1-\frac{1}{k^2}}$  } \\
\cline{1-4}
\multicolumn{1}{ |c  }{\multirow{2}{*}{$\mathcal{RED}$} }
&\multicolumn{1}{ |c| }{$s\sim cn$}
&\multicolumn{2}{|c|}{$s\sim cnk$} \\
\cline{2-4}
\multicolumn{1}{ |c  }{}
&\multicolumn{1}{ |c| }{$c_0=1$}
&\multicolumn{2}{|c|}{$c_0=\frac{(1+\sqrt{k+1})^2}{k(k-1)}$} \\
\cline{1-4}
\multicolumn{1}{ |c  }{\multirow{2}{*}{$\mathcal{RLN}$} }
&\multicolumn{1}{ |c| }{$s\sim cnk$}
&\multicolumn{2}{|c|}{$s$ fixed} \\
\cline{2-4}
\multicolumn{1}{ |c  }{}
&\multicolumn{1}{ |c| }{$c_0=(8/3\pi)^2$ \, $[n=k]$}
&\multicolumn{2}{|c|}{$c_0=k^2$} \\
\cline{1-4}\\
\end{tabular}
\end{center}
\caption {Thresholds for the sets of separable states $\mathcal{SEP}$, PPT states $\mathcal{PPT}$, states satisfying the reduction criterion $\mathcal{RED}$, and states satisfying the realignment criterion $\mathcal{RLN}$.} 
\label{tbl:thresholds-criteria}
\end{table}

Finally, in \cite{lancien2015extendibility}, Lancien studies the performance of $r$-extendibility criteria for random quantum states. Recall that a bipartite quantum state $\rho_{AB} \in \mathcal D_{nk}$ is said to be \emph{$r$-extendible} if there exists a $(r+1)$-partite state $\sigma_{AB^r} \in \mathcal D_{nk^r}$ which is invariant under all permutations of the $B$-systems and has $\rho_{AB}$ as a marginal:
$$[\mathrm{id}_{nk} \otimes \mathrm{Tr}_{k^{r-1}}](\sigma_{AB^r}) = \rho_{AB}.$$
Obviously, any separable state $\rho_{AB}$ is $r$-extendible, for all $r \geq 1$. Doherty, Parrilo, and Spedalieri have shown in \cite{doherty2004complete} that these conditions are also sufficient.
\begin{theorem}
A bipartite quantum state $\rho_{AB} \in \mathcal D_{nk}$ is separable if and only if it is $r$-extendible with respect to the system $B$ for all $r \in \mathbb N$.
\end{theorem}
In \cite{lancien2015extendibility}, besides computing estimates on the average width of the set of $r$-extendible states, Lancien computes a lower bound for the threshold value of these sets, for fixed $r$. 
\begin{proposition}\cite[Theorem 6.4]{lancien2015extendibility}
\label{prop:extendibility}
Fix $r \geq 1$ and consider balanced random quantum states $\rho_n \in \mathcal D_{n^2}$ having distribution $\nu_{n^2,s_n}$. For any $\varepsilon>0$, if the function $s_n$ is asymptotically smaller than $(1-\varepsilon)(r-1)^2/(4r)n^2$ as $n \to \infty$, then, 
$$\lim_{n \to \infty} \mathbb P[\rho_n \text{ is $r$-extendible}] = 0.$$
In other words, the threshold $c_0$ for the set of $r$-extendible states in the scaling $s \sim cn^2$ is larger than $(r-1)^2/(4r)$. 
\end{proposition}

Notice that the analysis in \cite{lancien2015extendibility} does not give any upper-bounds on the threshold $c_0$, see Question \ref{qst:extendibility}.

\subsection{Absolute separability}

Whether a quantum state $\rho$ is separable or entangled does not only depend on the spectrum of $\rho$: there are, for example, rank one (pure) states which are separable ($\rho = ee^* \otimes ff^*$) and other states which are entangled ($\rho = \Omega$, see \eqref{eq:maximally-entangled-state}). In other words, the separability/entanglement of a quantum state depends also on its eigenvectors. In order to eliminate this dependence, in \cite{kus2001geometry} the authors introduced the set of \emph{absolutely separable states}
$$\mathcal{ASEP}_{n,k} = \bigcap_{U \in \mathcal U_{nk}} U \cdot \mathcal{SEP}_{n,k} \cdot U^* = \{\rho \, : \, U\rho U^* \text{ is separable }\forall U \in \mathcal U_{nk}\} \subset \mathcal D_{nk}.$$
Obviously, the truth value of $\rho \in \mathcal{ASEP}_{n,k}$ depends only on the  spectrum $\lambda$ of the density operator $\rho$, so one could simply use
$$\Delta_{nk} \ni \widetilde{\mathcal{ASEP}}_{n,k} = \{ \lambda \, : \, \mathrm{diag}(\lambda) \in \mathcal{ASEP}_{n,k} \}.$$
Similarly, one can define absolute versions (and the corresponding spectral variants) for the sets $\mathcal{PPT}$, $\mathcal{RED}$, and $\mathcal{RLN}$.

An explicit description of the set $\mathcal{APPT}$ has been obtained in \cite{hildebrand2007positive}, as a finite set of positive-semidefinite conditions. The analogue question for $\mathcal{ARED}$ has been settled in \cite{jivulescu2015positive}, whereas the problem of finding an explicit description of the set $\mathcal{ARLN}$ remains open. Interestingly, it has been shown in \cite{johnston2013separability} that for qubit-qudit systems ($\min(n,k)=2$), absolute separability is equivalent to the absolute PPT property. Later, in \cite{arunachalam2014absolute} evidence towards the general conjecture $\mathcal{ASEP}_{n,k} = \mathcal{APPT}_{n,k}$ (for all $n,k$) has been collected; in particular, the authors show that for all $n,k$, $\mathcal{APPT}_{n,k} \subseteq \mathcal{ARLN}_{n,k}$. 

At the level of thresholds, the values (and even the scales) for $\mathcal{ASEP}$ and $\mathcal{ARLN}$ are completely open. The following results, for $\mathcal{APPT}$ and $\mathcal{ARED}$ are from \cite{collins2012absolute}, and respectively \cite{jivulescu2015thresholds}.

\begin{proposition}
The thresholds for the sets $\mathcal{APPT}_{n,k}$ are as follows:
\begin{enumerate}
\item In the balanced regime, where $n \geq k \to \infty$, the threshold value for the parameter $s$ of the induced measure $\nu_{nk,s}$ is on the scale $s \sim cnk^3$ at the value $c_0=4$.
\item In the unbalanced regime, where $n \to \infty$ and $k$ is fixed, the threshold value for the parameter $s$ of the induced measure $\nu_{nk,s}$ is on the scale $s \sim cnk$ at the value $c_0=(k+\sqrt{k^2-1})^2$.
\end{enumerate}

The thresholds for the sets $\mathcal{ARED}_{n,k}$ are as follows:
\begin{enumerate}
\item In the balanced regime, where $n, k \to \infty$, the threshold value for the parameter $s$ of the induced measure $\nu_{nk,s}$ is on the scale $s \sim cnk$ at the value $c_0=1$.
\item In the first unbalanced regime, where $k \to \infty$ and $n$ is fixed, the threshold value for the parameter $s$ of the induced measure $\nu_{nk,s}$ is on the scale $s \sim ck$ at the value $c_0=n-2$.
\item In the second unbalanced regime, where $n \to \infty$ and $k$ is fixed, the threshold value for the parameter $s$ of the induced measure $\nu_{nk,s}$ is on the scale $s \sim cnk$ at the value 
$$c_0=\left(1+\frac{2}{k}+\frac{2}{k}\sqrt{k+1}\right)^2.$$
\end{enumerate}

\end{proposition}

\section{Deterministic input through random quantum channels}
\label{sec:single-output}

Although a global understanding of the 
typical properties
of a random channel is desirable (and this is
the object of Section \ref{sec:random-output}),
obtaining results for interesting outputs
of given random channels is of intrinsic interest.
For example, as we explain subsequently in Section \ref{sec:conjugate-products}, 
the image of highly entangled states under
the tensor product of random channels
is an important question, as it 
is one of the keys to obtaining 
violations for the 
additivity of the minimum output entropy.

Our first model is a one-channel model that
consists in considering
matrices $X_n$ which have a macroscopic scaling $\trace(X^p) \sim n \cdot  \phi(x^p)$, 
where $x$ is some non-commutative random variable.
In order to obtain states, we normalize:
\[\tilde X = \frac{X}{\trace X}.\]
Therefore, the moments of the output matrix $Z = \Phi(\tilde X)$ are given by
\[\E[\trace(Z^p)] = \E[\trace(\Phi(\tilde X)^p)] = \E\left[ \trace \frac{\Phi(X)^p}{(\trace X)^p}\right] = \frac{\E[\trace(\Phi(X)^p)]}{(\trace X)^p}.\]

We consider different asymptotic regimes for the integer parameters $n$ and $k$. 
It turns out that the computations in the case of the
regime $k$ fixed, $n \to \iy$ is more involved, and its understanding requires  free probabilistic tools. To an integer $k$ and a probability measure $\mu$, we associate the measure $\mu_{(k)}$ defined by
\[\mu_{(k)} = \left( 1 - \frac 1 k \right) \delta_0 + \frac 1 k \mu.\]

\begin{proposition}
The almost sure behavior of the output matrix $Z=\Phi(\tilde X)$ is given by:
\begin{enumerate}
\item[(I)] When $n$ is fixed and $k \to \iy$, $Z$ converges almost surely to the maximally mixed state
\[\rho_* = \frac{1}{n}\I_n.\]
\item[(II)] When $k$ is fixed and $n \to \iy$, the empirical spectral distribution of $\bar \mu k n Z$ converges to the probability measure $\nu = [\mu_{(k)}]^{\boxplus k^2}$, where $\boxplus$ denotes the free additive convolution operation, $\mu$ is the probability distribution of $x$ with respect to $\phi$: $\phi(x^p) = \int t^p \; d\mu(t)$ and $\bar \mu$ is the mean of $\mu$, $\bar \mu = \phi(x)$. 
\item[(III)] When $n,k \to \iy$ and  $k/n \to c$, the empirical spectral distribution of the matrix $nZ$ converges to the Dirac mass $\delta_1$.
\end{enumerate}
\end{proposition}

\section{Random quantum channels and their output sets}
\label{sec:random-output}

We do this section in the chronological order. 

\subsection{Early results on random unitary channels}

\subsubsection{Levy's lemma}

Some results are already available in order to quantify the entanglement of 
generic spaces in $ \Gr_{p_n}(\C^n\otimes \C^k)$.
The best result known so far is arguably the
 following theorem of Hayden, Leung and Winter in
\cite{hayden2006aspects}:

\begin{theorem}[Hayden, Leung, Winter, \cite{hayden2006aspects}, Theorem IV.1]
Let $A$ and $B$ be quantum systems of dimension $d_A$ and $d_B$
with $d_B\geq d_A\geq 3$ Let $0<\alpha <\log d_A$. Then there exists a subspace
$S\subset A\otimes B$ of dimension
$$d\sim d_Ad_B\frac{\Gamma\alpha^{2.5}}{(\log d_A)^{2.5}}$$

such that all states $x\in S$ have entanglement satisfying
$$H(\lambda(x))\geq \log d_A -\alpha -\beta,$$
where $\beta =d_A/(d_B\log 2)$ and $\Gamma = 1/1753$.
\end{theorem}

For large $d$, Aubrun \cite{aubrun2009almost} studies quantum 
channels on $\mathbb C^d$ obtained by 
selecting randomly $N$ 
independent Kraus operators according to a probability 
measure $\mu$ on the unitary group $U(d)$. 
He shows the following result:
\begin{theorem}
Consider a random unitary channel $M_N\to M_N$
obtained with $d$ iid Haar unitaries. 
For for $N>> d/\varepsilon ^2$, such 
a channel is $\varepsilon$-randomizing 
with high probability, 
i.e. it maps every state within distance $\varepsilon /d$
of the maximally mixed state.
\end{theorem}

This slightly improves on the above result by Hayden, Leung, Shor and Winter by optimizing
their discretization argument.

\subsection{Results with a fixed output space}

We introduce now a norm on $\R^k$ which will have a very important role to play in the description of the set $K_{n,k,t}$ in the asymptotic limit $n \to \iy$. 

\begin{definition}\label{def:t-norm}
For a positive integer $k$, embed $\R^k$ as a self-adjoint real 
subalgebra $\mathcal R$ of a $\mathrm{II}_1$ factor  $\mathcal A$ endowed with trace $\phi,$ 
so that $\phi((x_1,\dots,x_k))=(x_1+\cdots+x_k)/k$. Let $p_t$ be a projection of rank $t \in (0,1]$ in $\mathcal A$, free from $\mathcal R$. On the real vector space $\R^k$, we introduce the following norm, called the \emph{$(t)$-norm}:
\begin{equation}
	\normt{x}:=\norm{p_t x p_t}_{\infty},
\end{equation}
where the vector $x \in \R^k$ is identified with its image in $\mathcal R$.
\end{definition}

We now introduce the convex body $K_{k,t}\subset \Delta_k$ as follows:
\begin{equation}\label{eq:convex}
	K_{k,t}:=\{ \lambda \in\Delta_{k} \; |\; \forall a\in\Delta_{k} , \scalar{\lambda}{a} \leq \normt a
	\},
\end{equation}
where $\scalar{\cdot}{\cdot}$ denotes the canonical scalar product in $\R^k$. 
We shall show later that this set is intimately related to the $(t)$-norm: $K_{k,t}$ is the intersection of the dual ball of the $(t)$-norm with the probability simplex $\Delta_k$. Since it is defined by duality, $K_{k,t}$ is the intersection of the probability simplex with the half-spaces
$$
H^+(a, t) = \{x \in \R^k \; | \; \scalar{x}{a} \leq \normt{a}\}
$$
for all directions $a \in \Delta_k$. Moreover, we shall show that every hyperplane $H(a, t) = \{x \in \R^k \; | \; \scalar{x}{a} = \normt{a}\}$ is a supporting hyperplane for $K_{k,t}$.

Let $(\Omega,\mathcal F,\P)$ be a probability space in which the sequence or random vector subspaces $(V_{n})_{n\geq 1}$ is defined.
Since we assume that the elements of this sequence are independent, we may assume that
$\Omega = \prod_{n\geq 1}\Gr_N(\C^k\otimes \C^n)$
and $\P=\otimes_{n\geq 1}\mu_{n}$ where $\mu_{n}$ is the invariant measure on the Grassmann manifold $\Gr_N(\C^k\otimes \C^n)$.
Let $P_{n} \in \M_{nk}(\C)$ be the random orthogonal projection whose image is $V_{n}$. For two positive sequences $(a_n)_n$ and $(b_n)_n$, we write $a_n \ll b_n$ iff $a_n / b_n \to 0$ as $n \to \iy$.
\begin{proposition}\label{prop:set-of-proba-one}
Let $\nu_{n}$ be a sequence of integers satisfying $\nu_{n}\ll n$.
Almost surely, the following holds true: 
for any self-adjoint matrix $A\in \M_{k}(\C)$, 
the $\nu_{n}$-th largest eigenvalues of $P_{n}(A\otimes I_{n})P_{n}$
converges to $||a||_{(t)}$
where $a$ is the eigenvalue vector of $A$. 
This convergence is uniform on any compact set of $\M_{k}(\C)_{sa}$.
\end{proposition}

\begin{proof}
For any self-adjoint $A\in \M_{k}(\C)$, the almost sure convergence follows from 
and from Theorem \ref{libre}.

Let $A_{l}$ be a countable family of self-adjoint matrices in $\M_{k}(\C)$ and assume that
their union is dense in the operator norm unit ball.
By sigma-additivity, 
the property to be proved holds almost-surely simultaneously for
all  $A_{l}$'s. 

This implies that the property holds for all $A$ almost-surely, as the $j$-th largest eigenvalue
of a random matrix is a 
Lipschitz function for the operator norm on the space of matrices. 
\end{proof}

The set on which the conclusion of the above proposition holds true will be denoted by $\Omega'$ and we therefore have $\P(\Omega ')=1$. Technically, $\Omega'$ depends on $\nu_{n}$ but in the proofs, we won't need to keep track of this dependence as  $\nu_n$ will be a fixed sequence.

The main result of our paper is the following characterization of the asymptotic behavior of the random set $K_{n,k,t}$. We show that this set converges, in a very strong sense, to the convex body $K_{k,t}$.

\begin{theorem}\label{thm:output-eigenvalues-single-channel}
\label{limite}
Almost surely, the following holds true:
\begin{itemize}
\item Let $\mathcal{O}$ be an open set in $\Delta_{k}$ containing $K_{k,t}$.
Then,  for $n$ large enough, $K_{n,k,t}\subset \mathcal{O}$.
\item Let $\mathcal K$ be a compact set in the interior of $K_{k,t}$.
Then, for $n$ large enough, $\mathcal K \subset K_{n,k,t}$.
\end{itemize}
\end{theorem}

\subsection{More results about the output of random channels}

More results are known about the output of random
quantum channels. 
Instead of giving a full list, let us state 
the following result from \cite{collins2013convergence}, that supersedes 
many results already known.

\begin{theorem}
Let $k$ be a fixed integer, and $\Phi_n : \mathcal M_n(\mathbb C)\to \mathcal M_k(\mathbb C)$ 
be a sequence of quantum channels
constructed with constant matrices and unitary
matrices that are independent from each other. 
Then, there exists a compact convex set $K$ such that its the random collection out output sets converges 
almost surely to $K$ in the topology induced by the Hausdorff distance
between compact sets. 
\end{theorem}

This theorem includes in particular encompasses
the following two important examples. 
Firstly, the random unitary channels
$$\tilde\Phi_n (X)=k^{-1}\sum U_iXU_i^*,$$
but also, more importantly 
a product 
$\chi_n=\Phi_n\otimes\Xi$, where $\Xi$ is any
quantum channel fixed in advance, 
and $\Phi_n$ is any of the sequences
considered previously. 

Actually, there is even more, namely: 
in the previous theorem, the 
output set $K$ can actually be exactly realized
via the collection of outputs of pure states
(no need for all input states). In addition,
the boundary of the collection of output sets
converges to the boundary of $K$ in the 
Hausdorff distance (which means that
any point in the interior of $K$ is attained
within finite time with probability one), and for 
any 
finite collection of $l$ elements in the interior of 
$K$, it is possible to find with probability 
one in finite time an family of pre-images by pure
states which are close to orthogonal to each other
(the tolerance is arbitrary and can be fixed
ahead of time). Somehow, this is the strongest
convergence one can hope for, and it is actually
rather counterintuitive that the image of the
extreme points of a convex body (the input states)
end up filling exactly the image of the convex body.

As a corollary, however, we obtain the following:

\begin{corollary}
In all examples of random channels taken so 
far, the Holevo capacity converges with probability 
one. In particular, if the image set $K$ contains
the identity, with probability one, 
$$\chi_{\Phi_n}+H_{min}(\Phi_n)\to \log k.$$\end{corollary}

\section{The additivity problem for tensor products of random quantum channels}
\label{sec:additivity}

\subsection{The classical capacity of quantum channels and the additivity question}

The following theorem summarizes some of the most important breakthroughs in quantum information
theory in the last decade. It is based in particular on the papers \cite{hastings2009superadditivity,hayden2008counterexamples}.

\begin{theorem}
For every $p \in [1, \infty]$, there exist
quantum channels $\Phi$ and $\Psi$ such that
\begin{equation}
H_p^{\min}(\Phi \otimes \Psi) < H_p^{\min}(\Phi) + H_p^{\min}(\Psi).
\end{equation}
\end{theorem}

Except for some particular cases ($p>4.73$, \cite{werner2002counterexample} and $p>2$, \cite{grudka2010constructive}), the proof of this theorem uses 
the random method, i.e. the channels $\Phi, \Psi$ are random channels, and the above inequality occurs with 
non-zero probability. At this moment, we are not aware of any explicit, non-random choices for $\Phi, \Psi$ in the case  $1 \leq p \leq 2$.

The additivity property for the minimum output entropy $H^{\min}(\cdot)$ was related in \cite{shor2004equivalence} to the additivity of another important entropic quantity, the \emph{Holevo quantity}
$$\chi(\Phi) = \max_{\{p_i,X_i\}} \left[ H_1\left(\sum_i p_i \Phi(X_i)\right ) -\sum_i p_iH_1( \Phi(X_i) ) \right].$$
The regularized Holevo quantity provides \cite{holevo1998capacity,schumacher1997sending} the classical capacity of a quantum channel $\Phi$, i.e.~the maximum rate at which classical information can be reliably sent through the noisy channel
$$.$$

\subsection{Conjugate quantum channels and the MOE of their tensor product}
\label{sec:conjugate-products}

In this subsection we gather some known results about the MOE of tensor products of conjugate channels $\Psi = \Phi \otimes \bar \Phi$. These results will be used in the next subsection on counterexamples. Let us stress from the beginning that in there is much less known about the output eigenvalues of $\Psi$ than about those of a single random channel $\Phi$. In particular, we do not have an explicit description of the output set of $\Psi$, such as the one from Theorem \ref{thm:output-eigenvalues-single-channel}. Actually, we have mostly upper bounds in this case, coming from the trivial inequality
\begin{equation}\label{eq:H-min-max-ent}
H^p_{\min}(\Psi) \leq H^p ([\Phi \otimes \bar \Phi](\Omega)),
\end{equation}
where $\Omega$ is the maximally entangled state \eqref{eq:maximally-entangled-state}.

The first result in this direction is a non-random one, giving a lower bound on the larges eigenvalue of the output of the maximally entangled state. To fix notation, let $\Phi : M_d(\mathbb C) \to M_k(\mathbb C)$ a quantum channel coming from an isometry $V:\mathbb C^d \to \mathbb C^k \otimes \mathbb C^n$. In \cite{hayden2008counterexamples},
the authors observed
that in the context of two random channels
given by two dilations $V_1$ (resp. $V_2$),
it is relevant to introduce the further symmetry $V_2 = \bar V_1$, as it ensures that at least one eigenvalue is always big. 

\begin{lemma}\label{lem:HW-trick}
The larges eigenvalue of the output state $Z = [\Phi \otimes \bar \Phi](\Omega_d)$ satisfies the following inequality:
$$\|Z\| \geq \frac{d}{nk}.$$
\end{lemma}
This result appeared several times in the literature (and it is sometimes referred to as the ``Hayden-Winter trick''), see \cite[Lemma 2.1]{hayden2008counterexamples} or \cite[Lemma 6.6]{collins2010random} for a proof using the graphical (non-random) calculus from Section \ref{sec:graphical-tensors}.

In the context of random quantum channels, one can improve on the result above, by computing the asymptotic spectrum of the output state $Z_n$. This has been done in \cite{collins2010random} in different asymptotic regimes. Since in this review we focus on the regime where $k$ is fixed and $d \sim tnk \to \infty$, we state next Theorem 6.3 from \cite{collins2010random}.

\begin{theorem}\label{thm:eigenvalues-output-2-conjugate}
Consider a sequence of random quantum channels coming from random isometries $\Phi_n:M_{d_n}
(\mathbb C) \to M_k(\mathbb C)$ where $d_n$ is a sequence of integers satisfying $d_n \sim tnk$ as $n \to \infty$ for fixed parameters $k \in \mathbb N$ and $t \in (0,1)$. The eigenvalues of the output state
$$M_{k^2}(\mathbb C) \ni Z_n = [\Phi_n \otimes \bar \Phi_n](\Omega_{d_n})$$
converge, almost surely as $n \to \infty$, to 
\begin{itemize}
\item $t+ \frac{1-t}{k^2}$, with multiplicity $1$;
\item $\frac{1-t}{k^2}$, with multiplicity $k^2-1$.
\end{itemize}
\end{theorem}
In order to prove such results, one uses the method of moments: using the Weingarten formula \eqref{eq:Weingarten} from Section \ref{sec:Weingarten}, it is shown in \cite[Section 6.1]{collins2010random} that, for all $p \geq 1$,
$$\frac{1}{k^2}\mathbb E \mathrm{Tr}(Z_n^p) = \sum_{\alpha,\beta \in \mathcal S_{2p}} n^{\#(\alpha^{-1}\gamma)}k^{-2 + \#\alpha} d_n^{\#(\beta^{-1}\delta)} \mathrm{Wg}_{nk}(\alpha,\beta).$$
where $\gamma, \delta$ are some fixed permutations in $\mathcal{S}_{2p}$; we present in Figure \ref{fig:Phi-bar-Phi} the diagram for the output matrix $Z_n$. We then compute the dominating terms in the above sums, by finding the pairs $(\alpha,\beta)$ corresponding to the terms having the largest $n$ powers; this is done by replacing $d_n = tkn + o(n)$ and using the asymptotic expression for the Weingarten factor from Theorem \ref{thm:mob}. It turns out that the set of dominating pairs $(\alpha,\beta)$ is small, and one can compute, up to $o(1)$ terms, the sum, proving the result. Since the matrices $Z_n$ live in a space of fixed dimension ($k^2$), a simple variance computation allows to go from the convergence in moments to the almost sure convergence of the individual eigenvalues.

Note that Theorem \ref{thm:eigenvalues-output-2-conjugate} improves on Lemma \ref{lem:HW-trick} in two ways: the norm of the output is larger, and we obtain information on the other eigenvalues too. This turns out to be useful in obtaining better numerical constants for the counterexamples to additivity, see the discussion in Section \ref{sec:counterexamples}.

\begin{figure}[htbp] 
\includegraphics{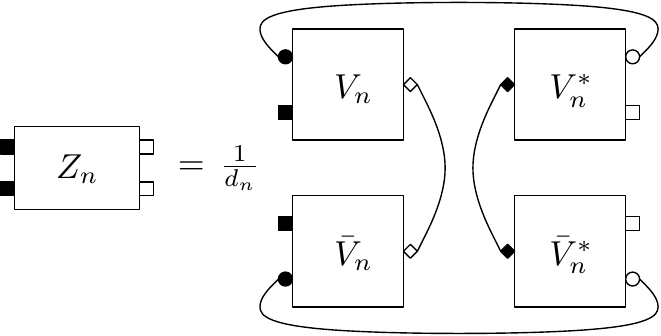}
\caption{Diagram for the output matrix $Z_n$, obtained by putting a maximally entangled state through a product of two conjugate quantum channels.} 
\label{fig:Phi-bar-Phi}
\end{figure}

Finally, the last result we would like to discuss in relation to products of conjugate channels is \cite[Theorem 5.2]{fukuda2014asymptotically}. The setting here is more general: the authors consider not one copy of a channel and its conjugate, but $2r$ channels (in what follows, $r$ is an arbitrary fixed positive integer):
$$\Psi_n = \Phi_n^{\otimes r} \otimes \bar \Phi_n^{\otimes r}.$$
Informally, \cite[Theorem 5.2]{fukuda2014asymptotically} states that, among a fairly large class of input states, the tensor products of Bell states ($\pi \in \mathcal S_r$ is an arbitrary permutation)
$$\Omega^\pi_{d_n} = \bigotimes_{i=1}^r \Omega_{d_n}^{i,\pi(i)}$$
are the ones producing outputs with least entropy. In the equation above, the maximally entangled state acts on the $i$-th copy of $\mathbb C^{d_n}$ corresponding to non-conjugate channels $\Phi_n$ and on the $\pi(i)$-th copy of $\mathbb C^{d_n}$ corresponding to conjugate channels $\bar \Phi_n$. The class of inputs among which the products of maximally entangled states are optimal are called ``well-behaved'', in the sense that they obey a random-matrix eigenvalue statistics; see \cite[eq.~(43)]{fukuda2014asymptotically} for more details.

The result above shows that inequality \eqref{eq:H-min-max-ent} is tight, when restricting the minimum on the left hand side to the class of well-behaved input states; the general question is open for random quantum channels, see Question \ref{qst:Bell-optimal}. Moreover, an important point raised in \cite{fukuda2014asymptotically} is that the optimality of maximally entangled inputs extends to tensor products of channels. This result might be useful for analyzing regularized versions of the minimum output entropies, in relation to the classical capacity problem.

\subsection{Early results in relation to the violation of MOE, history and the state of the art}
\label{sec:counterexamples}

We present next a short history of the various counterexamples to the additivity question, discussing different values of the parameter $p$. 

In the range $p>1$, the first counterexample was obtained by Werner and Holevo \cite{werner2002counterexample}: they have shown that the channel $(1-x)\mathrm{id} + x \mathrm{transp}$, acting on $\mathcal M_3(\mathbb C)$, for $x=-1/(d-1)$, violates the additivity of the $p$-R\'enyi entropy for all $p > 4.79$. Then, Hayden and Winter proved, in their seminal work \cite{hayden2008counterexamples}, that random quantum channels violate additivity with large probability, for all $p \geq 1$. The same result, using this time free probability, was obtained in \cite{collins2011random}, with smaller system dimensions. Also in the range $p>1$, Aubrun, Szarek and Werner proved violations of random channels, using this time Dvoretzky's theorem \cite{aubrun2010nonadditivity}. For $p$ close to $0$, violations of additivity were proved in \cite{cubitt2008counterexamples}.

The most important case, $p=1$, turned out to be much more difficult. The difficulty comes from the fact that one needs a precise control of the \emph{entire output spectrum}, while for $p>1$ controlling the largest eigenvalue turned out to be sufficient. The breakthrough was achieved by Hastings in \cite{hastings2009superadditivity}, where he showed that random mixed unitary channels violate additivity of the von Neumann entropy. Several authors, using similar techniques as Hastings, improved, generalized, and extended his result \cite{fukuda2010comments,brandao2010hastings, fukuda2010entanglement}. An improved version of Dvoretzky's theorem was used in \cite{aubrun2011hastingss} to show violations at $p=1$. Later, Fukuda provided a simpler proof of violation \cite{fukuda2014revisiting}, using this time $\varepsilon$-net arguments and Levy's lemma, the techniques used also in the pioneering work \cite{hayden2006aspects}. In \cite{belinschi2012eigenvectors,belinschi2013almost}, the authors use free probability theory to compute \emph{exactly} the minimum output entropy of a random quantum channel \cite[Theorem 5.2]{belinschi2013almost}. These results lead to the largest value of the violation known to date ($1$ bit), and the smallest output dimension ($k=183$), see Theorem \ref{thm:violation-p-1} below. 

Finally, let us mention that the majority of the violation results above use random constructions. The exceptions are the results in \cite{werner2002counterexample} ($p > 4.79$) and \cite{grudka2010constructive} ($p>2$, using the antisymmetric subspace); the question of finding other explicit counterexamples is open to this day, see Question \ref{qst:non-random-counterexamples}.

We state next the best result to date concerning violations of additivity for the minimum output entropy  \cite[Theorem 6.3]{belinschi2013almost}.

\begin{theorem}\label{thm:violation-p-1}
Consider a sequence $\Phi_n : \mathcal M_{\lfloor tkn \rfloor}(\mathbb C) \to \mathcal M_k(\mathbb C)$ of random quantum channels, obtained from random isometries
$$V_n:\mathbb C^{\lfloor tkn \rfloor} \to \mathbb C^k \otimes \mathbb C^n.$$
For any output dimension $k \geq 183$, in the limit $n \to \infty$, there exist values of the parameter $t \in (0,1)$ such that almost all random quantum channels violate the additivity of the von Neumann minimum output entropy. For any $\varepsilon$, there are large enough values of $k$ such that the violation can be made larger than $1$ bit. 

Moreover, in the same asymptotic regime, for all $k < 183$, the von Neumann entropy of the output state $[\Phi_n \otimes \bar \Phi_n](\Omega_{\lfloor tkn \rfloor})$ is almost surely larger than $2 H^{\min}(\Phi_n)$. Hence, in this case, one can not exhibit violations of the additivity using the Bell state \eqref{eq:maximally-entangled-state} as an input for the product of conjugate random quantum channels.
\end{theorem}

The above theorem leaves open the maximal possible value of the violation for conjugate random quantum channels, due to the fact that the maximally entangled state is not known to be optimal in this scenario, see Question \ref{qst:Bell-optimal}.

\section{Other applications of RMT
to quantum spin chains volume laws}
\label{sec:other}

\subsection{Maximum entropy principle for random matrix product states}
\label{sec:maximum-entropy-random-MPS}

Random matrix techniques play other roles
in quantum spin chain theory. In this section
we follow \cite{collins2013matrix}.

In the theory of quantum spin chains,  it is nowadays widely well justified, both numerically \cite{white1992density} and analytically \cite{hastings2007entropy}, that  ground states can be represented by the set of Matrix Product States with {\it polynomial} bond dimension. 
In the situation of a chain with boundary effects in exponentially small regions of size $b$ at both ends, homogeneity in the bulk and experimental access to an exponentially small central region of size $l$.
Tracing out the boundary terms leads to a bulk state given by \begin{equation}\label{eq:inicio}
\rho=\sum_{i_{b+1},...i_{N-b},j_{b+1},...j_{N-b}=1}^d tr (L A_{i_{b+1}} \cdots A_{i_{N-b}} R  A^*_{j_{N-b}} \cdots A^*_{j_{b+1}} )\ket {i_{b+1}...i_{N-b}} \bra {j_{b+1}...j_{N-b}},
\end{equation}
where all $A_i$, $L\ge 0$ and $R\ge 0$ are $D\times D$ matrices with $D={\rm poly}(N)$. 

In other words, the prior information can be
understood as restricting the bulk-states of our system as having the form (\ref{eq:inicio}).

It is known from the general theory of MPS \cite{perez-garcia2006matrix} that this set has a natural (over)para\-metrization by the group $U(dD)$, 
via the map $U\mapsto A_i=\bra 0 U \ket i$. 
In $U(dD)$, one can use the symmetry-based assignment of prior distributions to sample from the Haar measure. Similarly, the fact that the map $X\mapsto \sum_i A_iXA_i^*$ is trace-preserving leads to consider $tr(R)=1$, $\|L\|_\infty\le 1$, and gives natural ways of sampling also the boundary conditions (see below).
Looking for the generic reduced density matrix $\rho_l$ of $l\ll N$ sites
then becomes a natural problem. 
It corresponds to asking about generic observations of 1D quantum systems. This idea has been already exploited for the non-translational invariant case in \cite{garnerone2010typicality}. The aim of the present work is to show that $\rho_l$ has generically maximum entropy:

\begin{theorem}
Let $\rho_l$ be taken at random from the ensemble introduced with $D\geq N^{1/5}$. Then
$\| \rho_l/\mathrm{Tr} \rho_l-1/d^l\|_{\infty} \leq (d^l-1)\sqrt{d^l}O(D^{-1/10})$  except with probability exponentially small in D.
\end{theorem}

Note that, since the accessible region $l$ is exponentially smaller than the system size, the bound can be made arbitrary small while keeping the size of the matrices $D$ polynomial in the system size.

To prove the theorem, one needs the graphical Weingarten calculus provided in \cite{collins2010random} (see Sections \ref{sec:graphical-tensors} and \ref{sec:graphical-Wg-wick}) 
and an uniform estimate of the Weingarten function, more subtle than the one stated in Theorem \ref{thm:mob}.

Finally, in the same context of condensed matter physics, let us mention the work of Edelman and Movassagh, containing applications of random matrix theory and free probability to the study of the eigenvalue distribution of quantum many body systems having generic interactions \cite{movassagh2011density}.

\subsection{Multiplicative bounds for random quantum channels}

Once the additivity questions for the minimum $p$-R\'enyi entropy of random quantum channels had been settled in \cite{hayden2008counterexamples} and \cite{hastings2009superadditivity}, the attention shifted towards the \emph{amount} of the possible violations of the minimum output entropy. In \cite{montanaro2013weak}, Montanaro shows that random quantum channels are not very far from being additive by bounding the minimum output $\infty$-R\'enyi entropy of a tensor power of a channel by the same quantity for one copy of the channel. His idea is to bound the desired entropy by a \emph{additive quantity}, the norm of the partial transposition of the projection on the image subspace of the random isometry defining the channel. The following theorem is a restatement of \cite[Theorem 3]{montanaro2013weak}. 

\begin{theorem}
Let $\Phi: \mathcal M_d(\mathbb C) \to \mathcal M_k(\mathbb C)$ be a random quantum channel having ancilla dimension $n$. Suppose $k\leq n$, $\min \{d,k\} \geq 2 (\log_2 n)^{3/2}$ and $d=o(kn)$. 
Then, for any $p>1$, with high probability as $n \to \infty$, the following inequality holds
$\frac 1 r H^{\min}_p (\Phi^{\otimes r}) \geq \beta(1-1/p) H^{\min}_p (\Phi) $ where
$$\beta \simeq \begin{cases} 1/2 & \text{if } d \geq n/k\\ 
1 & \text{if } d \leq n/k\end{cases}.$$
\end{theorem} 

Soon after, Montanaro's ideas were pursued in \cite{fukuda2015additivity}. There, different additive quantities (e.g.~the operator norm of the partial transpose of the Choi matrix of the quantum channel) were used to bound the minimum output $2$-R\'enyi entropy. The results provide slight improvements, in the case of interest $p=1$ over the bounds from \cite{montanaro2013weak}. The following statement follows from \cite[Theorem 8.4]{fukuda2015additivity}.

\begin{theorem} Consider a sequence of random quantum channels $\Phi_n:\mathcal M_d(\mathbb C) \to \mathcal M_k(\mathbb C)$ with ancilla dimension $n$, where $k$ is a fixed parameter and $d \sim tnk$ for a fixed $t \in (0,1)$. Then, almost surely as $n \to \infty$, for all $p\in[0,2]$, there exist constants  $\alpha_p \in [0,1]$ such that, for all $r \geq 1$,
\begin{equation}
\frac 1 r H_p^{\min} (\Phi_n^{\otimes r}) \geq  \alpha_p H_p^{\min}(\Phi_n).
\end{equation}
The constants $\alpha_p$ satisfy the following relations
\begin{enumerate}
\item When $0<t<1/2$ is a constant,
$$ \alpha_p = o(1) + \frac{p-1}{2p} \left [ 1 + \frac {2 \log 2 + \log (1-t)}{\log t} \right] \cdot \mathbf{1}_{(1,2]}(p).$$
\item When $k$ is large and $t \asymp k^{-\tau}$ with $\tau >0$,
$$ \alpha^\Gamma_{p,k,t} = o(1) +  \begin{cases}
 \frac{p-1}{2p} &\quad \text{ if }  0< \tau \leq 1-1/p\\
 \tau  /2 &\quad \text{ if } 1-1/p \leq \tau \leq 2\\
1 &\quad \text{ if }  \tau \geq 2.
\end{cases} $$
\end{enumerate}
\end{theorem}

Incidentally, since the limiting spectrum of the partial transposition of the Choi matrix is computed in \cite{fukuda2015additivity}, the authors show the existence of PPT quantum channels violating generically the additivity of the minimum $p$-R\'enyi entropy, for all $p \geq 30.95$, see \cite[Theorem 10.5]{fukuda2015additivity}.

\subsection{Sum of random projections on tensor products}

Ambainis, Harrow and Hastings \cite{ambainis2012random}
consider a problem in random matrix theory that is inspired by quantum information theory:
determining the largest eigenvalue of a sum of p random product states 
in $(\C^d){\otimes k}$
 where $k$ and $p/d^k$
are fixed while $d\to\infty$. When $k = 1$, the Mar{\v{c}}enko-Pastur law determines asymptotically the largest eigenvalue 
$(1+\sqrt{p/d^k})^2$, the smallest eigenvalue, and 
the spectral density. 

More precisely, their setup is as follows:
for each dimension $d$, let $(p_d^{(i)})_{i\in \{1,\ldots , k\} }$
be independent uniformly distributed rank one random 
projections on $\C^d$. 
\begin{theorem}
As $d\to\infty$, the operator norm of
$$\sum_i p_d^{(1)}\otimes \cdots \otimes p_d^{(k)}$$
still behaves almost surely like
$(1+\sqrt{p/d^k})^2$ and the spectral density approaches that of the
Mar{\v{c}}enko-Pastur law \eqref{eq:Marchenko-Pastur}.
\end{theorem}

Their proof is essentially based on moment
methods. 
Direct computation of moments of high order 
allow to conclude. Various methods are proposed
by the author, including methods of Schwinger-Dyson
type. It would be interesting to 
see whether these methods that are well established
in theoretical physics and random matrix
theory could be of further use in quantum
information theory. 
This result generalizes  the random matrix theory result to the random tensor case, and for the
records, this is arguably one of the first
precise results about the convergence of 
norms of sums of tensor products when 
the dimensions of each legs are the same. 

The original motivation of the authors
emanates in part from problems
 related to quantum data-hiding. 
 We refer to \cite{ambainis2012random}
 for the proofs and motivations.

\subsection{Area laws for random quantum states associated to graphs}
\label{sec:area-laws}

We would like to generalize now Proposition \ref{prop:area-law-adapted} to the general case of non-adapted marginals. The theorem in this section makes use of random matrix theory techniques, more precisely is build on the moment computation done in \cite[Theorem 5.4]{collins2010randoma}.

Before we state the area law, we need to properly define the \emph{boundary} of a the marginal induced by a partition $\{S,T\}$ of the total Hilbert space. In the adapted marginal case discussed in Section \ref{sec:random-states-graphs}, this definition was natural; the general situation described here requires a preliminary optimization procedure. 

To keep things simple, assume that all local Hilbert spaces have the same dimension $N$. A partition $\{S,T\}$ defines, at each vertex of the graph, a pair of non-negative integers $(s(v), t(v))$ such that $s(v) + t(v) = \mathrm{deg}(v)$ and $\sum_v s(v) = |S|$, $\sum_v t(v) = |T|$. The randomness in the unitary operators $U_v$ acting on the vertices of $G$ introduces an ``incertitude'' on the choice of the copies of $\mathbb C^N$ which should be traced out at each vertex $v \in G$. The following definition of the boundary volume removes this incertitude by performing an optimization over all possible choices for the partial trace. Note that the case of adapted marginals (see Definition \ref{def:boundary-volume-adapted}) does not require this optimization step, since there is no incertitude (at each vertex, either all or none of the subsystems are traced out).

\begin{definition}\label{def:area-boundary}
For a graph $G$ and a marginal $\rho_S$ of the graphs state $\varphi_G$ defined by a partition $\{S,T\}$, define the \emph{boundary volume} of the partition as 
$$|\partial S| = \max_\alpha \mathrm{cr}(\alpha),$$
where $\alpha$ is a function $\alpha : [2m] \to \{S,T\}$ defining which copies of $\mathbb C^n$ are traced out, and $\mathrm{cr}(\alpha)$ is the number of \emph{crossings} in the assignment $\alpha$, that is the number of edges in $G$ having one vertex in $\alpha^{-1}(S)$ and the other one in $\alpha^{-1}(T)$. 
\end{definition}

The following theorem is the main result of \cite{collins2013area}, showing that the area law holds for random graph states, with the appropriate definition of the boundary volume. Moreover, one can compute the correction term to the area law, a quantity which depends on the topology of the graph $G$. We refer the interested reader to \cite[Sections 5,6]{collins2013area} for the definition of the correction term $h_{G,S}$ and the proofs.

\begin{theorem}
Let $\rho_S$ be the marginal $\{S, T\}$ of a graph state $\varphi_G$. Then, as $N \to \iy$, the  \emph{area law} holds, in the following sense
\begin{equation}\label{eq:area-law}
	\E H(\rho_S) = |\partial S| \log N - h_{\Gamma, S} + o(1),
\end{equation}
where $|\partial S|$ is the area of the boundary of the partition $\{S,T\}$ from Definition \ref{def:area-boundary} and $h_{\Gamma, S}$ is a positive constant, depending on the topology of the graph $G$ and on the partition $\{S,T\}$ (and independent of $N$).
\end{theorem}

\section{Conclusions and open questions}
\label{sec:questions}

We finish this review article with
a series of questions that seem to be 
of interest at the intersection
of random matrix related techniques, 
and quantum information / quantum mechanics. 

In relation to the various threshold result from Section \ref{sec:thresholds}, we list next several important open questions. 

\begin{question}\label{qst:threshold-sep}
Is it possible to remove the $\log$ factors from Theorem \ref{thm:threshold-sep} and to obtain a sharper threshold result for the set $\mathcal{SEP}$ of separable states?
\end{question}

Regarding the hierarchy of $r$-extendibility criteria, the upper bound corresponding to the threshold result in \ref{prop:extendibility} is open, see \cite[Section 9.2]{lancien2015extendibility}. 

\begin{question}\label{qst:extendibility}
Find a constant $c_1 \geq (r-1)^2/(4r)$ such that random quantum states having distribution $\nu_{n^2, c_1n^2}$ are, with high probability as $n \to \infty$, $r$-extendible. 
\end{question}

Regarding the random entanglement criteria introduced in Theorem \ref{thm:k-positivity-from-mu}, one can define
$$K_{\mu,m}=\{ \rho \in \mathcal D_{nm} \, : \, [f_d\otimes \mathrm{id}_m] (\rho) > 0 \text{ almost surely, for $d$ large enough}\}.$$
The following question, addressing the global power of such random criteria, was left open in \cite{collins2015random}.
\begin{question}\label{qst:random-entanglement-criteria}
Define the set of quantum states satisfying \emph{all} random criteria from Theorem \ref{thm:k-positivity-from-mu}
$$K^{free}_{n,k,m} := \bigcap_{\mu \, : \, \mathrm{supp} \left( \mu^{\boxplus n/k} \right) \subset [0, \infty)} K_{\mu , m}.$$
Can one give an analytical description of $K^{free}_{n,k,m}$? It was shown in \cite[Proposition 3.7]{collins2015random} that the only pure states contained in $K^{free}_{n,k,m}$ are the separable (product) ones. Are there values of the parameters $n,k,m$ for which the set $K^{free}_{n,k,m}$ is precisely the set of $k$-separable states from $\mathcal D_{nm}$?
\end{question}

In Section \ref{sec:maximum-entropy-random-MPS}, we have discussed a random model for matrix product states, and we have shown it obeys the maximum entropy principle of Jaynes. There are also natural questions related 
to quantum spin chains: 
\begin{question}
In Section \ref{sec:area-laws},
we stated an area law for random quantum states. 
Given a random Hamiltonian $H_N$ acting on 
$\C^N$, let $H^{(i)}$ be the operator 
obtained from $H_N$ acting 
on $(\C^N)^{\otimes k}$ by the action of
$H_N$ on the $i$th leg, and identity elsewhere. 
We assume that we come up with a model
with a gap, i.e. the difference between 
its smallest eigenvalue and its second
smallest eigenvalue is uniform. 
It follows from results by 
Hastings \cite{hastings2007entropy}
that the ground state of the Hamiltonian 
$\sum H_N^{(i)}$ satisfies an area law. 
If $H_N$ has some randomness in addition,
can we obtain more precise results, e.g.
regarding the distribution of the ground
state? 
In the same vein, can random techniques
allow us to obtain results for other
topologies, e.g. in the 2D context?
\end{question}

Let us now consider some open questions in quantum information theory, related to random matrices. 

As discussed in Section \ref{sec:random-states}, there are several ways in which one can define random quantum states. All classes of probability measures discussed in Section \ref{sec:random-states} are very well motivated, both from the mathematical and the physical standpoints. In \cite{nechita2012random}, the authors introduce a new ensemble of random quantum states, by considering iterations of random quantum channels. The following question was asked in \cite[Section 4]{nechita2012random}. 

\begin{question}
Compute the statistics of the probability measure $\nu_{b}$ on the set of quantum states $\mathcal D_n$ defined as follows. For a probability vector $b \in \Delta_k$, consider the quantum channel 
$$\Phi(X) = [\mathrm{id}_n \otimes \mathrm{Tr}_k](U(X \otimes \mathrm{diag}(b))U^*),$$
where $U \in \mathcal U_{nk}$ is a random Haar unitary. Then, $\nu_b$ is the probability distribution of the \emph{unique} invariant state of $\Phi$ (uniqueness is shown in \cite[Theorem 4.4]{nechita2012random}. 
\end{question}

Regarding the various counterexamples in the literature for the minimum output entropy and other capacity-related questions, we list next several open problems. 

\begin{question}\label{qst:non-random-counterexamples}
Construct explicit, non-random counterexamples to the additivity of the $p$-R\'enyi entropy, in the range $p \in [1,2]$. 
\end{question}

\begin{question}\label{qst:Bell-optimal}
Is the maximally entangled state $\Omega$ the actual minimizer of the minimum output entropy for a pair of conjugate random quantum channels $\Phi \otimes \bar \Phi$?
\end{question}

\begin{question}
Does a pair $\Phi \otimes \Psi$ of \emph{independent} random quantum channels violate additivity of the quantities $H_p^{\min}(\cdot)$? 
\end{question}

Regarding the known violations of the additivity of the MOE entropy for pairs of conjugate channels, it is important to note that Theorem \ref{thm:violation-p-1} only allows to obtain bounds on the output dimension of the random channels. Previous results (see, e.g.~\cite{fukuda2010comments}) allow to bound also the input dimension. The approach used in \cite{belinschi2012eigenvectors,belinschi2013almost}, using free probability, uses estimates of objects existing at the limit where the input dimension is infinity. It would thus be desirable, in this framweork, to be able to work at finite input dimension, and thus bound \emph{all} the relevant parameters which allow for additivity violations. 
\begin{question}
A random contraction is known to be determinantal \cite{metcalfe2013universality} 
and the determinant involves contour integrals.
So far, many random matrix techniques 
used for QIT rely either on concentration of
measure, or on moment methods. Is it possible
to use complex analysis methods
(steepest descent, Riemann-Hilbert problem
analysis) in order to refine existing estimates.
For example, can such estimates give 
bounds for dimensions of input spaces for 
violation of MOE?
\end{question}

Finally, we would like to end the current review with a very important open question, regarding different regularized quantities for random quantum channels. 
\begin{question}
Compute the almost sure limit of the regularized $H_p^{\min}(\cdot)$ quantities, the Holevo capacity and the classical capacity for random quantum channels.
\end{question}

\bigskip

\noindent \textit{Acknowledgements.}
B.C.'s research was partly supported by NSERC, ERA, and Kakenhi funding. I.N.'s research has been supported by a von Humboldt fellowship and by the ANR projects {OSQPI} {2011 BS01 008 01} and {RMTQIT}  {ANR-12-IS01-0001-01}. Both B.C. and I.N. were supported by the ANR project {STOQ}  {ANR-14-CE25-0003}.

\bibliography{review-rmt-jmp}{}
\bibliographystyle{alpha}

\end{document}